\documentclass[11pt]{article}

\usepackage[a4paper,left=25mm,right=20mm,top=25mm,bottom=20mm]{geometry}
\usepackage[T1]{fontenc}
\usepackage[utf8]{inputenc}
\usepackage{lmodern}
\usepackage{setspace}
\usepackage{amsmath,amssymb,amsthm,bm}
\usepackage{booktabs}
\usepackage{makecell}
\usepackage{multirow}
\usepackage{subcaption}
\usepackage{algorithm}
\usepackage{algpseudocode}
\usepackage{tikz}
\usepackage{graphicx}
\usepackage{float}
\usepackage{url}
\usepackage{natbib}
\usepackage[hidelinks]{hyperref}
\newtheorem{theorem}{Theorem}
\newtheorem{lemma}[theorem]{Lemma}
\newtheorem{corollary}[theorem]{Corollary}
\theoremstyle{remark}
\newtheorem{remark}[theorem]{Remark}

\providecommand{\BIBand}{and}
\bibpunct[, ]{(}{)}{,}{a}{}{,}

\title{Solving Hard Instances from Knapsack and Bounded Knapsack Problems: A new state-of-the-art solver}

\author{
Renan F. F. da Silva\\
Institute of Computing, University of Campinas\\
\texttt{renan.silva@students.ic.unicamp.br}
\and
Thiago Alves de Queiroz\\
Instituto de Matem\'atica e Estat\'istica, Universidade Federal de Goi\'as\\
\texttt{taq@ufg.br}
\and
Rafael C. S. Schouery\\
Institute of Computing, University of Campinas\\
\texttt{rafael@ic.unicamp.br}
}

\date{}

\begin{document}

\maketitle

\begin{abstract}
The \emph{Knapsack Problem} (KP) and its generalization, the \emph{Bounded Knapsack Problem} (BKP), are classical NP-hard problems with numerous practical applications. Although the solvers \texttt{COMBO} and \texttt{BOUKNAP} were introduced more than 25 years ago, they remain the state-of-the-art for KP and BKP, respectively, due to their highly optimized implementations and sophisticated bounding techniques. In this work, we present \texttt{RECORD} (\emph{\underline{RE}fined \underline{COR}e-based \underline{D}ynamic programming}), a new solver for both problems that builds upon key components of \texttt{COMBO}---including core- and state-based dynamic programming, weak upper bounds, and surrogate relaxation with cardinality constraints---while introducing novel strategies to address its limitations. Specifically, we propose \emph{multiplicity reduction} to limit the number of distinct item types, combined with on-the-fly item aggregation, refined fixing-by-dominance techniques, and a new divisibility bound that strengthens item fixing and symmetry breaking. These enhancements enable our solver to retain \texttt{COMBO}'s near-linear-time behavior on most instances while achieving substantial speedups on more challenging cases. Computational experiments demonstrate that \texttt{RECORD} consistently outperforms \texttt{COMBO} and \texttt{BOUKNAP} on difficult benchmark sets, including recent instances from the literature, often by several orders of magnitude, thereby establishing a new state-of-the-art solver for both problems.
\end{abstract}


\noindent\textbf{Funding:} This research was supported by the State of Goi\'as Research Foundation (FAPEG); the National Council for Scientific and Technological Development (CNPq), grant numbers 408722/2023-1, 315555/2023-8, 444679/2024-3, 407005/2024-2, 312345/2023-2, and 404779/2025-5; and the S\~ao Paulo Research Foundation (FAPESP), grants \#2023/17964-4 and \#2022/05803-3.

\section{Introduction}
In the \emph{Knapsack Problem} (KP), also known as the \emph{Binary Knapsack Problem}, we are given a set of items \( I = \{1, \ldots, n\} \), where each item \( i \in I \) is characterized by a profit \( p_i \), a weight \( w_i \), and an efficiency \( p_i / w_i \), together with a knapsack of capacity \( W \). The objective is to select a subset \( I' \subseteq I \) such that the total weight of the selected items does not exceed the knapsack capacity, i.e., \( \sum_{i \in I'} w_i \leq W \), while maximizing the total profit \( \sum_{i \in I'} p_i \). In the \emph{Bounded Knapsack Problem} (BKP), each item \( i \in I \) is additionally associated with an availability \( d_i \), indicating that up to \( d_i \) copies of item \( i \) may be selected.

The KP is a classical and widely studied NP-hard problem in combinatorial optimization. Despite its simple binary formulation with a single linear constraint, it appears as a subproblem in numerous important applications. For instance, given $D$ dimensions, in the \emph{multidimensional (or vector) knapsack problem}, each item has a weight \( w_i^d \) in each dimension $d = 1,\ldots, D$, which makes it a general integer linear program with binary variables and positive coefficients \citep{Cacchiani_2022}. In the \emph{multiple knapsack problem} \citep{Pisinger_1999, Lalonde_2022}, instead of a single knapsack, several knapsacks are given, potentially with different capacities. The \emph{bin packing problem} (BPP) considers a set of items with associated weights and an unlimited number of bins of identical capacity, with the objective of minimizing the number of bins required to pack all items \citep{Wei_2020, Baldacci_2024, Silva_2025}. The set partitioning formulation solved via a branch-and-price algorithm is a successful approach for solving the BPP, in which the KP arises as the pricing problem. The KP is also a fundamental subproblem in \emph{capacitated vehicle routing problems}, in which vehicles with limited capacity must serve clients with heterogeneous demands \citep{Pessoa_2020}.
A relevant extension of the KP is the \emph{Knapsack Problem with Conflicts} (KPC) \citep{Coniglio_2021}, in which compatibility constraints prohibit certain pairs of items from being selected simultaneously. When such conflict constraints are imposed within the BPP, the resulting problem is the \emph{Bin Packing Problem with Conflicts} \citep{Sadykov_2013, Wei_2020}, which constitutes a generalization of the classical \emph{Vertex Coloring Problem} \citep{Morrison_2014}.
Additional classes of constraints have been investigated in the BPP literature, including precedence constraints, where bins are indexed by time slots, and constraints are imposed on pairs of items to enforce minimum and/or maximum temporal separations \citep{Letelier_2022}. Within branch-and-price frameworks, these precedence constraints give rise to pricing problems that generalize the KPC.

Two books have been dedicated exclusively to the KP and its variants. The first, by~\citet{Martello_1990}, presents around 200 results obtained over the preceding thirty years. The second, by~\citet{Kellerer_2004}, provides a comprehensive overview of the substantial advances achieved during the subsequent thirteen years, citing more than 500 references. In addition,~\citet{Cacchiani_2022} presented a recent survey on knapsack problems, reviewing the latest developments on the classical and related single-knapsack variants. The authors discussed decades of progress and highlighted the continued relevance of this simple yet challenging problem.

Although numerous studies have been published over the past few decades, the fastest known exact solvers for the KP and BKP remain \texttt{COMBO}~\citep{Martello_1999} and \texttt{BOUKNAP}~\citep{Pisinger_2000}, respectively. Remarkably, these solvers were proposed more than 25 years ago, and no improvements have been reported since then. This longevity can be attributed to the solvers' meticulous design and highly optimized implementations, which allow them to solve certain instances in linear time. In particular, \texttt{COMBO} incorporates several low time-complexity features, making many large-scale instances relatively easy to solve. Moreover, the implementation of \texttt{COMBO} goes well beyond what is typically found in optimization software, where the primary focus is often on demonstrating ideas. \texttt{COMBO} places a strong emphasis on low-level optimization, minimizing memory allocation, optimizing cache usage, and maintaining small constant factors across its components. As a result, both the strength of its features and the quality of its implementation leave little room for further improvement. Any new solver aiming to surpass \texttt{COMBO} must therefore adhere to an equally high standard of code efficiency, as even a mediocre implementation of features with the same asymptotic complexity can lead to computing times that are orders of magnitude slower. This explains, in part, the lack of substantial progress over the past 25 years.

Several techniques employed by these solvers were proposed by previous works. A predecessor of \texttt{COMBO} is the \texttt{MINKNAP} algorithm~\citep{Pisinger_1997} that applies a core-based dynamic programming~(DP) approach combined with linear programming~(LP) bounds to prune states and fix items. \texttt{COMBO} builds upon \texttt{MINKNAP} by integrating additional techniques from the literature, such as surrogate relaxation with cardinality constraints (a similar Lagrangian approach was introduced previously by~\citet{Martello_1997}), a divisibility bound that tightens the capacity using the greatest common divisor of item weights, and a pairing heuristic for generating incumbent solutions. In contrast, \texttt{BOUKNAP} extends \texttt{MINKNAP} to the BKP by dealing with item availabilities, which makes some fixing techniques more efficient, although it does not include all the enhancements implemented in \texttt{COMBO}.

Furthermore, both \texttt{BOUKNAP} and \texttt{COMBO} are so efficient that, for most benchmark instances, the most time-consuming step is often the sorting procedure, which runs in \(O(n\log n)\) and sorts items by non-increasing efficiency. In some cases, they are even faster, as both solvers employ a partial-sorting strategy based on a median-finding method that partitions items into subintervals sorted relative to each other, performing full sorting only on the necessary intervals. This optimization yields near-linear runtimes on many instances.

Nonetheless, some instances remain challenging, forcing the algorithms to operate close to their worst-case time complexities of
\( O(nW) \) and \( O\!\left(W \sum_{i \in I} \log_2 d_i\right) \), respectively.
These bounds stem from the DP nature of both algorithms, which may require the enumeration of up to that many distinct states. We next discuss such instances, which, although originally proposed for the KP, can be adapted to the BKP by grouping identical items, without significantly reducing their computational difficulty.

Several difficult instances were proposed by \citet{Pisinger_2005}, with the number of items ranging from \(20\) to \(10^4\).
The Pisinger benchmark combines classes from earlier works, extended with larger coefficients (i.e., weights, profits, and capacities), as well as new classes with distinct structural properties that remain challenging even for small coefficients.
While most instances are relatively easy, a small subset poses difficulties for \texttt{COMBO}, requiring several seconds to solve, and the hardest instances may take several minutes.

In particular, we observe that Pisinger instances with large coefficients for which \texttt{COMBO} exhibits long running times often satisfy $p_i = w_i + C$ for some fixed $C \in \mathbb{Z}$. A closer examination of these instances indicates that the difficulty arises when the solution process reduces to a challenging subset-sum problem, either in the original formulation or in integer surrogate relaxations.

On the other hand, Pisinger instances with small coefficients and long runtimes exhibit two main structures. The first consists of classes generated from a \emph{span} of items, where a small base set generates all others via integer multiples, introducing strong symmetry as many subsets share identical total weights and profits. The second corresponds to the \emph{profit ceiling} class, in which weights are random and profits are defined as \( p_i = d \lceil w_i / d \rceil \) for a given integer \( d \). These instances can be difficult when linear programming bounds are weak due to divisibility issues.

More recently, \citet{Miles_2021} proposed a more diverse benchmark using instance-space analysis, while preserving a difficulty level comparable to the Pisinger benchmark. Furthermore, \citet{Jooken_2022} introduced a large benchmark of 3{,}240 hard KP instances with fewer items, ranging from 400 to 1{,}200. These instances induce extremely high runtime and memory requirements for \texttt{COMBO}, with memory consumption often reaching several tens of gigabytes. In these instances, item efficiencies are close to one with small noise, and items are designed to generate a very large number of near-equivalent solutions with slightly different profits and weights. As a result, we observed that the current solvers must maintain approximately \( W/10 \) non-dominated states, leading to extremely long runtimes for large knapsack capacities such as \( W = 10^{10} \). Indeed, in the experiments of \citet{Jooken_2022}, all 263 instances not solved within the two-hour time limit have this knapsack capacity. This contrasts with the Pisinger benchmark, where large capacities also occur, but the number of non-dominated states is typically only a small fraction of \( W \), preventing runtimes from reaching several hours.

In this work, we present a new exact solver based on DP for both the KP and BKP. Our approach leverages key features from \texttt{COMBO}, including its core DP technique, LP upper bounds, and surrogate relaxation with cardinality constraints. An interesting observation is that our algorithm, as well as several prior methods such as \texttt{COMBO} and \texttt{BOUKNAP}, can be viewed as combinatorial branch-and-bound methods, where the DP states correspond to branch-and-bound nodes. These nodes can be pruned either by LP bounds computed via combinatorial techniques or by fixing-by-dominance mechanisms arising from the DP substructure.

The short runtimes of \texttt{COMBO} on many instances can be attributed to the above-mentioned features, which either allow most items to be fixed through bounding techniques or keep the set of non-dominated states very small. However, \texttt{COMBO} struggles on several instances due to symmetries induced by items with identical efficiency values, as well as the presence of dominated or highly similar items that cannot be eliminated by bounds. These characteristics lead to a large number of non-dominated states and, consequently, long runtimes. In addition, weak LP bounds caused by divisibility issues further degrade performance.

To overcome these limitations and address the most challenging instances, our solver incorporates several new techniques. These include \emph{multiplicity reduction}, which decreases the number of distinct items; on-the-fly item aggregation; refined fixing-by-dominance rules; and a new divisibility bound. Collectively, these techniques enable more aggressive item fixing, reduce the size of the state space, and improve the effectiveness of fixing strategies by breaking symmetries. In addition, we introduce new heuristic procedures designed to rapidly obtain high-quality solutions for hard instances.

By integrating these reduction and bounding techniques, our solver preserves the near-linear efficiency of \texttt{COMBO} on most instances while significantly extending its effectiveness to more difficult instances, particularly those proposed by~\citet{Pisinger_2005}. Although our solver exhibits slightly worse performance on some easy instances due to the overhead introduced by the new features, computational experiments show that it is several times faster than \texttt{COMBO} on hard instances, including the recent benchmark set of~\citet{Jooken_2022}.

All proofs are provided in the appendices. The remainder of the paper is organized as follows. Sections~\ref{sec::core} and~\ref{sec::proposal} present the core-based DP algorithm and an overview of our approach, respectively. Section~\ref{sec::literature_reductions} reviews the upper-bounding techniques from the literature incorporated into our algorithm, including the linear relaxation, weak upper bounds, the state-based LP bound, and relaxations with cardinality constraints. Section~\ref{sec::our_reductions} introduces bounds on the maximum capacity and minimum profit of states, a guarded DP extension, our multiplicity reduction combined with on-the-fly item aggregation, a divisibility bound, fixing-by-dominance techniques, and our heuristics. Section~\ref{sec::experiments} reports the computational experiments, and Section~\ref{sec::conclusion} concludes the paper and discusses future research directions.

\section{The Core-based Dynamic Programming}~\label{sec::core}
We adopt the core-based DP approach implemented in~\texttt{BOUKNAP}. Below, we reinterpret its recurrence so that it resembles the classical recursive formulation of the BKP, which is more intuitive, and we use this reinterpretation throughout the description of our algorithm.

Without loss of generality, assume that items are sorted in non-increasing order of efficiency $e_i = p_i / w_i$. Let $b$ denote the break item, i.e., the index such that
$\sum_{i=1}^{b-1} d_i w_i \le W$ and $\sum_{i=1}^{b} d_i w_i > W$. Define modified profits and weights as $\overline{p}_i = -p_i$ and $\overline{w}_i = -w_i$ for $i < b$, and $\overline{p}_i = p_i$, $\overline{w}_i = w_i$ otherwise. We call the \emph{break solution} the solution defined by the first $b -1$ items, including all their availabilities, whose value and weight are $\widehat{p} = \sum_{i=1}^{b-1} d_i p_i$ and $\widehat{w} = \sum_{i=1}^{b-1} d_i w_i$. We may assume that $\widehat{w} < W$; otherwise, the knapsack is completely filled with the most efficient items, implying optimality.

As observed by~\citet{Pisinger_1995}, optimal solutions can often be obtained by removing a small number of \emph{left items} ($i < b$) and adding a small number of \emph{right items} ($i \ge b$). Motivated by this observation, consider the following recurrence that uses the break solution as a pre-assigned solution and is defined over an arbitrary permutation $A = \{A_1,\ldots, A_n\}$ of $\{1,\ldots,n\}$:
\[
dp_A(i,w) =
\max_{\substack{0 \le k \le d_{A[i]} \\ 0 \le w - k\,\overline{w}_{A[i]} \le 2W}}
\left\{
dp_A(i-1,\, w - k\,\overline{w}_{A[i]}) + k\,\overline{p}_{A[i]}
\right\}.
\]
The base case is $dp_A(0,w) = \widehat{p}$ for $\widehat{w} \le w \le 2W$, and $dp_A(0,w) = -\infty$ for $0 \le w < \widehat{w}$. The recurrence can be solved in time $O\!\left(\sum_{i=1}^n d_i W\right)$, and the optimal value is $dp_A(n,W)$. Observe that if instead we take the pre-assigned solution as empty (i.e., $\widehat{p}=\widehat{w}=0$), the recurrence reduces to the classical textbook formulation of the BKP.

Several upper-bounding and item-fixing techniques (see, e.g.,~\citet{Pisinger_1997,Pisinger_2000}) achieve their best performance when the permutation $A$ is structured so that left items ($i<b$) appear in non-increasing order of efficiency and right items ($i\ge b$) in non-decreasing order; we refer to any permutation satisfying these properties as a \textbf{good permutation}. A canonical example is $A' = \{b-1,\, b,\, b-2,\, b+1,\, b-3,\, b+2,\, \ldots\}$. When our recurrence is evaluated under this canonical permutation, it becomes mathematically equivalent to the \emph{core-based dynamic programming} approach described in~\citet{Pisinger_1995,Pisinger_2000,Martello_1999}. For a position $j$ in $A'$, define its \textbf{core} as the interval $[l_j,r_j]$, where $l_j$ is the smallest left index and $r_j$ the largest right index appearing in the prefix $A'[1,j]$; if no left (resp.\ right) item appears, set $l_j=b$ (resp.\ $r_j=b-1$). The name of the algorithm stems from the fact that solving the recurrence for indices $1,\ldots,j$ is equivalent to solving the instance restricted to the core $[l_j,r_j]$, under the assumption that all items with $i<l_j$ are included and all items with $i>r_j$ are excluded.

As noted by~\citet{Pisinger_1995}, fully sorting items by efficiency may be a bottleneck for many easy instances, since optimality is often achieved by modifying only a few items of the break solution. Therefore, state-of-the-art algorithms adopt \emph{lazy sorting}~\citep{Pisinger_1995,Pisinger_1997,Pisinger_2000,Martello_1999}, constructing the permutation $A^*$ on the fly as $A^* = \{\, l^1,\, r^1,\, d^1,\, l^2,\, d^2,\, r^2,\, \ldots \,\}$, where $l^k < l^{k+1}$, $r^k > r^{k+1}$, and $d^k$ denotes items excluded from enumeration by fixing techniques. The construction of $A^*$ in our algorithm is detailed in Appendix~\ref{sec::permutation}.

The recurrence \( dp_A \) can be solved by enumerating every cell \( (i, w) \); however, better performance can be achieved through a state-based enumeration. A state \( s \) is characterized by a profit \( s_p \), a weight \( s_w \), and auxiliary information that enables the recovery of the solution represented by the state (e.g., recording that \( s \) was obtained from another state \( s' \) by adding a copy of item \( i \)). For simplicity, we henceforth represent a state \( s \) solely by the pair \( (s_p, s_w) \), except when explicitly stated otherwise, as in Appendix~\ref{sec::recover}, which describes our solution recovery procedure. Since \( s_w \) may exceed \( W \), such a state does not necessarily correspond to a feasible solution. In any case, it is important to note that the enumeration of certain infeasible states \( (s_p, s_w) \) with \( W < s_w \leq W + \widehat{w} \) may be necessary to obtain an optimal solution.

We solve the recurrence \( dp_A \) using a state-based bottom-up approach combined with several fixing techniques. For each index \( j \) from 1 to \( n \), we compute the states \( S_j \) by enumerating all items in the subarray \( A[1, j] \). Specifically, we consider the states derived from the base state \( (\widehat{p}, \widehat{w}) \) by evaluating all possibilities of either removing left items or including right items with indices in the interval \(A[1, j] \), according to their respective available number of copies. A state \( (s_p, s_w) \in S_j \) is said to be dominated if there exists another state \( (s_p', s_w') \in S_j \) such that \( s_p' \geq s_p \) and \( s_w' \leq s_w \). Since every state derived from \( (s_p, s_w) \) or \( (s_p', s_w') \) is obtained by enumerating the subarray \( A[j+1, n] \), all states derived from \( (s_p, s_w) \) are dominated by those derived from \( (s_p', s_w') \). Therefore, at iteration $j$, it suffices to compute the set of non-dominated states $S_j$ obtained by enumerating the subarray $A[1,j]$. If two states mutually dominate each other, one must still be retained.

\begin{remark}
The above process can be interpreted as a branch-and-bound algorithm in which each state \( s \in S_j \) corresponds to a node where all variables associated with items in \( A[1,j] \) are fixed to specific values, while others associated with items in \( A[j+1,n] \) remain unfixed. This interpretation enables the use of combinatorial upper bounds to prune nodes, that is, to eliminate states from \( S_j \) that cannot yield a solution better than the incumbent, i.e., the best-known feasible solution. Moreover, these upper bounds can indicate which items in \( A[j+1,n] \) may or may not contribute to transforming a state \( s \in S_j \) into an improved solution.
\end{remark}

We keep the states in $S_j$ sorted by increasing weight, and we start with \( S_0 = \{(\widehat{p}, \widehat{w})\} \). If the item \( A[j] \) is discarded by the fixing technique, then \( S_j = S_{j-1} \). On the other hand, to extend \( S_{j-1} \) to \( S_j \), we perform up to \( k \leq d_{A[j]} \) iterations. Initially, we set \( S_j^0 = S_{j-1} \). At iteration \( k \), for each state \( (s_p, s_w) \in S_j^k \), we create a new state \( (s_p + \overline{p}_{A[j]}, s_w + \overline{w}_{A[j]}) \), provided that \( 0 \leq s_w + \overline{w}_{A[j]} \leq 2W \). The newly generated states can be merged with \( S_j^k \) in time complexity \( O(|S_j^k|) \), discarding states pruned by dominance or bound, and forming a new set \( S_j^{k+1} \) ordered by increasing weight. After \( k \) iterations, the resulting non-dominated states \( S_j^{k+1} \) correspond to \( S_j \).

\begin{remark}~\label{rmk::decomposition}
The latter procedure may be enhanced by reducing the \( d_{A[j]} \) iterations to a logarithmic number through the use of \emph{binary decomposition}. Notice that the \( d_{A[j]} \) copies of item \( A[j] \) can be represented as a set of buckets with multiplicities \( \mathcal{M} = \{1, 2, 4, \ldots, 2^h, a \}\), where \( h \) is the largest integer such that \( \sum_{i = 0}^h 2^i \leq d_{A[j]} \) and $a = d_{A[j]} - \sum_{i = 0}^h 2^i$. The last bucket can be discarded when $a = 0$. This way, any integer quantity \( 0 \leq k \leq d_{A[j]} \) can be represented as a subset sum of these buckets. By iterating once for each bucket, we perform \( \lceil \log_2(d_{A[j]} + 1) \rceil \) iterations. This procedure results in a DP algorithm with worst-case time complexity \( \mathcal{O}\left( W \sum_{i \in I} \log_2 d_i \right) \), the same as that found in~\texttt{BOUKNAP}.
\end{remark}

\section{Proposal Overview}~\label{sec::proposal}
As mentioned in the previous section, our algorithm adopts a core-based DP approach. We enhance this framework by incorporating several additional strategies: (i) heuristics that identify high-quality incumbent solutions at an early stage, thereby reducing the number of nodes explored; (ii) fixing procedures based on LP relaxations and divisibility bounds; (iii) fixing-by-dominance mechanisms that allow items to be skipped during enumeration or limit the number of copies that need to be considered; (iv) item aggregation and multiplicity reduction, which group identical items and reduce the number of distinct items, respectively, thereby mitigating symmetries in the branch-and-bound algorithm; and (v) surrogate relaxations used to strengthen upper bounds.

Algorithm~\ref{main_algorithm} provides an overview of the proposed method. In lines~1 and~2, the algorithm begins by applying lazy sorting to determine \( b \), computing the break solution, and executing the initial heuristics. Line~3 initializes the first set of states \( S_0 \), which contains only the break solution, as well as the flags \( flag_{\text{LSR}} \) and \( flag_{\text{AMR}} \), indicating whether the linear surrogate relaxation and the item aggregation/multiplicity reduction are still pending (i.e., not yet applied). Line~4 initializes the current index \( j \) of the permutation vector \( A \) to \( 0 \) (note that index \( 0 \) is invalid since \( A \) is 1-indexed), together with the counters \( c_{\text{LSR}} \), \( c_{\text{HP}} \), and \( c_{\text{DIV}} \), which control the execution of the linear surrogate relaxation, the heavy primal heuristics, and the divisibility bounds, respectively. In addition, line~4 computes the corresponding thresholds \( T_{\text{LSR}} \), \( T_{\text{HP}} \), \( T_{\text{DIV}} \), and \( P_{\text{SR}} \); the first three determine when each mechanism is triggered, while \( P_{\text{SR}} \) triggers both the linear and integer surrogate relaxations.

The permutation \( A \), or more specifically \( A^* \), is generated on the fly through repeated calls to lazy sorting, and our pseudocode handles it implicitly by requesting the next unfixed item index greater than a given index. An item may have its availability partially or completely fixed by the fixing techniques, resulting in a reduced unfixed availability \( u_i \), where \( u_i = 0 \) if the item is completely fixed and \( u_i = d_i \) if none of its copies is fixed.

\algtext*{EndIf}
\algtext*{EndFor}
\algtext*{EndWhile}
\algtext*{EndRepeat}
\algtext*{EndLoop}
\algtext*{EndProcedure}
\algtext*{EndFunction}
\algrenewcommand{\algorithmicrequire}{\textbf{Input:}}
\begin{algorithm}[H]
\small
\setlength{\baselineskip}{0.6\baselineskip}
\caption{RECORD}
\begin{algorithmic}[1]
\Require items $I$ and capacity $W$
\State Find break item $b$ by using lazy sorting and compute break solution $(\widehat{p},\widehat{w})$
\State Call initial heuristics.
\State $S_0 \gets \{(\widehat{p},\widehat{w})\}$; \;  $flag_{LSR}, flag_{AMR} \gets true$
\State $j, c_{LSR}, c_{HP}, c_{DIV} \gets 0$ and compute the thresholds $T_{LSR}, T_{HP}, T_{DIV}, P_{SR}$
\While{$j \le |A|$}
    \State $j' \gets j$
    \State $j \gets$ next $j''>j$ with $A[j'']$ not fixed
    \If{$j > |A|$} \textbf{break} \EndIf
    \State $i \gets A[j]$; \; $S_j^0 \gets S_{j'}$
    \State Obtain unfixed availability $u_i \le d_i$; decompose $u_i$ into buckets $\mathcal{M} = \{m_0,\ldots,m_t\}$
    \For{$k \gets 0$ to $t$}
        \State Extend $S_j^k \to S_j^{k+1}$ using $m_k$ copies of $i$
        \If{dominance allows skip subsequent iterations} \textbf{break} \EndIf
    \EndFor
    \State $S_j \gets S_j^{k+1}$ and update counters $c_{LSR}, c_{HP}, c_{DIV}$ by adding $|S_j|$
    \If{incumbent improved} try to enhance it via the greedy completion heuristic \EndIf
    \If{$|S_j| \geq n$} apply fixing-by-dominance to find items that can be fixed \EndIf
    \State{Run sampling pairing heuristic}

    \If{$flag_{LSR}$ and ($c_{LSR}\!\ge\!T_{LSR}$ or $|S_j|>P_{SR}$)} call linear surrogate relaxation; $flag_{LSR}\!\gets\!false$ \EndIf
    \If{$|S_j|>P_{SR}$ and there are candidates} call integer surrogate relaxation; \EndIf
    \If{$c_{HP}\!\ge\!T_{HP}$} run heavy primal heuristics; $c_{HP}\!\gets\!0$; $T_{HP}\!\gets\! 2 T_{HP}$ \EndIf
    \If{$c_{DIV}\!\ge\!T_{DIV}$}
        \State Call divisibility bounds; $c_{DIV}\!\gets\!0$; $T_{DIV}\!\gets\! 2 T_{DIV}$
        \If{$flag_{AMR}$}
            \State Sort remaining items; $flag_{AMR}\!\gets\!false$
            \State Call item aggregation and multiplicity reduction
        \EndIf
    \EndIf
\EndWhile
\State \Return incumbent (optimal solution)
\end{algorithmic}~\label{main_algorithm}
\end{algorithm}

While the current index \( j \le |A| \), the algorithm proceeds as follows. Line~6 stores the previous index \( j \) as \( j' \), and line~7 identifies the next index \( j \) corresponding to a non-fixed item; the procedure terminates in line~8 if no such index exists. Otherwise, the current item \( i = A[j] \) is selected, and the initial state set \( S_j^0 \) is inherited from the previous iteration ($S_{j'})$. Lines~11--13 iterate over the buckets of the binary decomposition of \( u_i \), applying fixing-by-dominance to possibly skip the remaining iterations for item \( i \).

In line~14, the resulting set \( S_j^{k+1} \) becomes the new state set \( S_j \), and all counters are incremented by \( |S_j| \). Lines~15--16 attempt to improve the incumbent solution if it has been recently updated and apply additional dominance tests to skip unprocessed items whenever \( |S_j| \ge n \). Line~17 executes the sampling pairing heuristic, while lines~18--19 trigger the linear and integer surrogate relaxations, respectively. The integer surrogate relaxations correspond to modified BKP instances generated by the linear version (referred to as ``candidates'') and are solved by a variant of \texttt{RECORD} for which surrogate relaxations and multiplicity reductions are disabled. Lines~20--22 execute the heavy primal heuristics and divisibility bounds, with their thresholds doubled after each execution. Additionally, after the first application of the divisibility bounds, lines~23--25 fully sort the remaining items, thereby completely constructing \( A \), and apply item aggregation and multiplicity reduction.

\section{Upper Bounds}~\label{sec::literature_reductions}
We present several upper bounds from the literature that are incorporated into our algorithm, including the LP and surrogate relaxations, weak LP upper bounds, and the state-based LP bound.

\subsection{Linear Relaxation}
An integer formulation for the BKP is as follows:
\begin{align}
    \text{maximize} \quad & \sum_{i = 1}^{n} p_i x_i &&~\label{eq:bkp-obj} \\
    \text{subject to} \quad & \sum_{i = 1}^{n} w_i x_i \leq W &&~\label{eq:bkp-cap} \\
    & x_i \leq d_i && \text{for all } i \in \{1, \ldots, n\} &&~\label{eq:bkp-bounds} \\
    & x_i \in \mathbb{Z}_+ && \text{for all } i \in \{1, \ldots, n\} &&~\label{eq:bkp-int}
\end{align}
By relaxing the integrality constraints, we obtain a formulation of the \emph{Fractional Bounded Knapsack Problem} (FBKP), which provides an LP relaxation of the BKP. Since items are sorted in non-increasing order of efficiency \( e_i \), an optimal FBKP solution $x^{\text{LP}}$ takes the form:
\begin{equation}
\label{eq:optimal_fract}
x_i^{\text{LP}} = d_i \quad \text{for } i < b, \quad
x_i^{\text{LP}} = 0 \quad \text{for } i > b, \quad
x_b^{\text{LP}} = \left( W - \sum_{i=1}^{b-1} d_i w_i \right) / w_b.
\end{equation}

Given an incumbent solution value $z$, LP relaxations help with upper bounds and may identify items that must be fixed for any better solution (i.e., with value at least $z+1$). We define an \emph{unfixed availability vector} \( u \in \mathbb{Z}_+^n \), representing the range of copies of each item that an improved solution can use. Specifically, for $i < b$, an improved solution contains between $d_i - u_i$ and $d_i$ copies of item $i$; while for $i \ge b$, $u_i$ denotes the maximum number of copies of $i$ that can appear in an improved solution (possibly $u_i < d_i$).

\subsection{Weak Upper Bound}~\label{sec::weak_upper}
We use the \emph{weak upper bound} of~\citet{Dembo_1980} to compute tight \emph{unfixed availabilities} \(u_i\). For each item \(i\), the bound
\[
WB(i,e)=\widehat{p}+e\overline{p}_i+\left(W-\widehat{w}-e\overline{w}_i\right)\frac{p_b}{w_b}
\]
gives an upper bound on the FBKP value when \(e\) copies of item \(i\) are removed (\(i<b\)) or added (\(i\ge b\)). Given the incumbent value $z$, if \(WB(i,e)<z+1\), no improving solution exists removing/adding $e$ copies of $i$. Since \(WB(i,e)\) is monotone in \(e\), the largest feasible \(e\) determines a tight availability \(u_i\). Using determinant notation \(\det(a_1,a_2,a_3,a_4)=a_1a_3-a_2a_4\), \citet{Pisinger_2000} showed that \(u_i\) can be computed in constant time as
\begin{equation}
u_i = \Bigl\lfloor \frac{\det(z+1 - \widehat{p}, W - \widehat{w}, p_b, w_b)} {\Delta_i} \Bigr\rfloor, \quad\text{where}\quad \Delta_i =
\begin{cases}
-\det(p_i, w_i, p_b, w_b), & i < b,\\
\det(p_i, w_i, p_b, w_b), & i \ge b.
\end{cases} \label{eq:reduction} \end{equation}
If \(p_i/w_i=p_b/w_b\) or \(u_i>d_i\), we set \(u_i=d_i\). The resulting vector \(u\) replaces \(d\) in the core-based DP, significantly reducing the search space.

\subsection{The State-based LP bound}
The state-based LP bound is defined as follows. After performing the \( k \)-th iteration for item \( A[j] \), we check, for each state \( s \in S_j^{k+1} \), whether it can be pruned using the LP bound \eqref{eq:state_bound}. Let \( nl \) and \( nr \) denote the next items to the left and to the right that will be processed in the next iteration of a left and right item, respectively. If \( A[j] \) is a left item and \( k \) is not the last iteration for \( j \), then \( nl = A[j] \). Otherwise, \(nl\) is the next unfixed left item in the subarray \(A[j+1, n]\), which can be found by a call to the sorting algorithm. The item \( nr \) is defined analogously. If there is no left or right item, we define \( nl = 0 \) and \( nr = n+1 \), respectively, and set \( p_0 = \infty \), \( p_{n+1} = 0 \), and \( w_0 = w_{n+1} = 1 \).

The LP bound for a state \(s = (s_p, s_w)\) is defined as:
\begin{equation}
\label{eq:state_bound}
B(s) =
\begin{cases}
s_p + (W - s_w)\dfrac{p_{nr}}{w_{nr}}, & \text{if } s_w \le W,\\[6pt]
s_p + (W - s_w)\dfrac{p_{nl}}{w_{nl}}, & \text{otherwise}.
\end{cases}
\end{equation}
Notice that $B(s)$ provides an upper bound on the best integer solution value reachable from state $s$ in subsequent iterations. This holds because all items in the range $[nl + 1, nr - 1]$ have already been fully enumerated, and we may now remove items with index $i \le nl$ or add items with index $i \ge nr$. Since the items are sorted by efficiency, \(B(s)\) represents an upper bound on the optimal FBKP solution, and consequently also on the best integer solution value derived from state \(s\).

It is important to mention that if there is no left item, then any state \( s \) with \( s_w > W \) is infeasible, and \( B(s) = -\infty \). Furthermore, if \( B(s) < z + 1 \), then state \( s \) can be pruned, since no improved solution can be derived from \( s \).

\subsection{Surrogate Relaxation with Cardinality Constraints}~\label{sec::sr}
Although the state-based LP bound is typically tight for the BKP, certain instances exhibit weak gaps, particularly when every optimal integer solution contains exactly $K$ items while the LP relaxation selects a fractional number $K' \in (K, K+1)$.

Following~\citet{Martello_1997}, we incorporate cardinality information by introducing lower and upper bounds $N_{\min}$ and $N_{\max}$ on the number of items in any improving solution. Since enforcing these constraints directly within the DP would result in a multi-constraint knapsack problem, we instead strengthen the formulation via \emph{Surrogate Relaxation} (SR). We provide a brief overview below. A detailed exposition is presented in Appendix~\ref{ap::sr}, including a rationale for preferring SR over the more widely studied \emph{Lagrangian relaxation} in this setting.

In the SR, the capacity and cardinality constraints are aggregated using a multiplier $\mu$, yielding a Surrogate Bounded Knapsack Problem (SBKP). Relaxing integrality leads to a Surrogate Fractional Bounded Knapsack Problem (SFBKP), for which the multiplier producing the tightest upper bound $UB^{SF}$ can be computed via binary search. If $UB^{SF} \geq z + 1$, we solve the corresponding SBKP with the best integer multiplier. In practice, integer multipliers are sufficient to obtain bounds nearly identical to those produced by optimal real multipliers. To avoid numerical overflow in large instances while preserving strong bounds, we adopt a binary search with expanding intervals.

Additional refinements are applied in special cases. For instances satisfying $p_i = w_i + C$, for some fixed $C \in \mathbb{Z}$, the optimal multiplier is often $C$, which can be verified in linear time. In such cases, we observe that SFBKP typically provides a bound as strong as SBKP; thus, its execution is skipped.

Moreover, we notice that the FBKP solution~\eqref{eq:optimal_fract} may result in \(x^{\text{LP}}\) selecting
\[
    \Gamma^{\text{LP}} = \sum_{i=1}^n x^{\text{LP}}_i,
\]
with $N_{\min} \le \Gamma^{\text{LP}} \le N_{\max}$. Recall that \( \Gamma^{\text{LP}} \) is noninteger, otherwise, \( x^{\text{LP}} \) is an optimal integer solution. In the above case, adding cardinality constraints with $L = N_{\min}$ or $K = N_{\max}$ does not strengthen the bound. According to~\citet{Martello_1999}, we can instead solve two SBKPs obtained from enforcing $L = \lceil \Gamma^{\text{LP}} \rceil$ and
$K = \lfloor \Gamma^{\text{LP}} \rfloor$, resulting in two bounds $UB^{SF}_L$ and $UB^{SF}_K$.
The overall upper bound is the maximum of these two values, and an SBKP is generated whenever each bound exceeds the current incumbent $z$. In \texttt{COMBO}, this idea is applied when $\Gamma^{\text{LP}}$ differs from either $N_{\min}$ or $N_{\max}$ by less than~$1$, i.e.,
\[
    N_{\min} < \Gamma^{\text{LP}} < N_{\min} + 1
    \quad \text{or} \quad
    N_{\max} - 1 < \Gamma^{\text{LP}} < N_{\max}.
\]
In the preliminary computation experiments, we observed that it is also beneficial when the deviation of $\Gamma^{\text{LP}}$ from $N_{\min}$ or $N_{\max}$ is less than~$2$.

\section{Proposed Techniques and New Features}~\label{sec::our_reductions}
While the previous sections described how we build upon existing contributions and address known limitations, this section presents additional techniques that further improve the state-of-the-art. In particular, we introduce bounds on the maximum capacity and minimum profit of maintained states, a guarded DP expansion mechanism, a reduction based on multiplicity, a new divisibility bound, and fixing-by-dominance techniques.

\subsection{Maximum State Capacity and Minimum State Profit}\label{sec::max_cap}
We must maintain states with a weight of at most $W_{\max} = W + \widehat{w} < 2W$. Additionally, $W_{\max}$ can be reduced as left items are processed or fixed in the knapsack. In particular, if $k$ copies of item $i$ ($i < b$) are processed, or if $u_i$ is reduced by $k$, then $W_{\max}$ can be decreased by $k w_i$. Hence, a skipped iteration can be interpreted as a processed one, allowing a corresponding reduction of $W_{\max}$.

Analogously, we impose a minimum state profit and keep a state $s$ only if $s_p + p_{\text{right}} > z$, where $z$ is the incumbent value and $p_{\text{right}}$ is the sum of the profits of the unfixed availabilities of all unprocessed right items.

\subsection{Guarded DP-extension}
In some instances, long sequences of DP extensions generate no new states, although the LP upper bounds indicate that improvements may be possible. These iterations simply copy the current state set (possibly applying pruning), leading to unnecessary memory traffic. To mitigate this effect, we adopt a guarded expansion strategy. We maintain separate counters for consecutive left and right iterations that generate no new states. When a counter exceeds a threshold (initially 10), we execute a read-only check to verify whether the current item can generate at least one feasible new state before executing the expansion routine. Because this procedure does not perform memory writing, it is significantly more cache-friendly than an expansion call. If no state can be generated, the expansion is skipped.

If the check detects that a new state can be generated, then running the checker was a waste of time, since the expansion must be executed anyway. To reduce the likelihood of such wasted work, the threshold is doubled. Additionally, an expansion is enforced after 40 consecutive skips to allow state-pruning mechanisms to act periodically. In practice, this strategy significantly accelerates instances where the optimal solution has already been found but cannot yet be certified by bounding.

\subsection{Divisibility Bounds}~\label{sec::divisibility}
If divisibility issues are not explicitly handled by the algorithms, even simple instances, such as a subset-sum problem in which all items have even weights while the knapsack capacity is odd, may result in worst-case algorithmic behavior. A first technique to address this kind of issue is the \emph{trivial divisibility bound} incorporated by \texttt{COMBO} for the KP, in which the remaining knapsack capacity, denoted by \(\tilde{W}\), is obtained by subtracting from $W$ the weights of all item copies already fixed inside the solution. The remaining capacity \(\tilde{W}\) is then divided by the greatest common divisor (gcd) of the weights of all unfixed items. We propose a similar bound for the BKP as follows in Lemma~\ref{lm::trivial_div}.

\begin{lemma}~\label{lm::trivial_div}
    Let \(I_\text{left}\) be the subset of left items in \(I\), and let \(I_\text{res}\) be the subset of \(I\) with \(u_i > 0\). If an improved solution exists, then for each item \(i \in I_\text{left}\), \(d_i - u_i\) copies of \(i\) belong to it, and it has a total weight at most
    \[
    \overline{W} = \sum_{i \in I_\text{left}} (d_i - u_i) w_i + \left\lfloor \frac{\tilde{W}}{\gcd(I_\text{res})} \right\rfloor \gcd(I_\text{res}),
    \]
    where \(\tilde{W} = W - \sum_{i \in I_\text{left}} (d_i - u_i) w_i\) and \(\gcd(I_\text{res})\) denotes the greatest common divisor of the weights of all items in \(I_\text{res}\). Therefore, we can set the knapsack capacity $W$ equal to \(\overline{W}\).
\end{lemma}

Nevertheless, more sophisticated instances with divisibility issues still challenge \texttt{COMBO}, even when the bound in Lemma~\ref{lm::trivial_div} is applied, leading to large runtimes. We overcome such a situation by proposing a second divisibility bound in Theorem~\ref{th::div_rule2}, thereby saving computing time.

\begin{theorem}~\label{th::div_rule2}
Let \(I^{1}_{\text{left}} \subseteq I_{\text{left}}\) be the set of items \(i\) such that \(u_i = 1\).
Assume that every improved solution contains at least \(|I^{1}_{\text{left}}| - 1\) unfixed copies of items in \(I^{1}_{\text{left}}\).
If there exists an item \(h \in I^{1}_{\text{left}}\) such that, after reducing its availability by one, the optimal value of the resulting FBKP relaxation is strictly smaller than \(z + 1\), then we can set \(u_h = 0\) in the original instance without loss of optimality.
\end{theorem}

The technique described in Theorem~\ref{th::div_rule2} handles the major divisibility issues observed in the benchmark instances proposed by~\citet{Pisinger_2005}.
At the same time, we notice two additional challenges that must be addressed to make it effective in practice.
The first one is that computing the exact optimal FBKP value for each item \(h \in I^{1}_{\text{left}}\) can be expensive. We instead compute an upper bound on such a value.

\begin{corollary}
Consider the modified BKP instance in which the demand for item \(h\) is decreased by one, and let \((\widehat{p}_h, \widehat{w}_h)\) denote the total profit and weight of item copies fixed inside the knapsack. For such an instance, let \(\overline{W}_h\) be the knapsack capacity given by Lemma~\ref{lm::trivial_div}, \(I_{\text{res}}\) be the set of remaining items with unfixed copies, and \(g\) be the item in \(I_{\text{res}}\) with the highest efficiency. The following is a valid upper bound for the modified BKP:
\[
    \mathrm{UB}_h = \widehat{p}_h + \bigl(\overline{W}_h - \widehat{w}_h \bigr)e_g.
\]
\end{corollary}
The upper bound $\mathrm{UB}_h$ is particularly strong for the instances under study, as it is often the case that \(e_g = e_b\). Moreover, it is much faster to compute, as we explain below. Whenever \(\mathrm{UB}_h < z + 1\), we set \(u_h = 0\) according to Theorem~\ref{th::div_rule2}. The second challenge is to identify when \(I^{1}_{\text{left}}\) satisfies the conditions of Theorem~\ref{th::div_rule2}, which we address in Theorem~\ref{thm:I2-pairwise}.

\begin{theorem}~\label{thm:I2-pairwise}
Suppose that for every item \(i \in I^{1}_{\text{left}}\),
\[
WB(i,1) \ge z + 1
\quad \text{and} \quad
WB(i,2) < z + 1~\text{hold}.
\]
Therefore, no two distinct items \(i, i' \in I^{1}_{\text{left}}\) can be removed simultaneously while maintaining a weak upper bound of at least \(z + 1\).
\end{theorem}

We refer to the technique in Theorem~\ref{th::div_rule2} as the \emph{Enhanced Divisibility Bound}. The two divisibility bounds run in time complexity \(O(n \log w_{\max})\), where \(w_{\max}\) denotes the maximum item weight. For the enhanced divisibility bound, the set \(I_{\text{res}}\) is the same across all modified BKP instances. The greatest common divisor of all items in \(I\) or \(I_{\text{res}}\) can be computed using \(O(n)\) calls to the Euclidean algorithm, each requiring time complexity \(O(\log w_{\max})\). After this preprocessing step, the capacities \(\overline{W}\) and \(\overline{W}_h\) can be computed in constant time, and consequently, the value \(\mathrm{UB}_h\) can also be computed in constant time for each item \(h\).

The trivial divisibility bound is triggered whenever \(|S_j| > 1000\), following the same strategy used in \texttt{COMBO}.
The enhanced divisibility bound is activated once a threshold \(T_{\text{DIV}}\) is exceeded, with initial value set to \(5 n \log n\), according to our preliminary experiments, and can be applied whenever all items in \(I^{1}_{\text{left}}\) satisfy the conditions of Theorem~\ref{thm:I2-pairwise}.

\subsection{Item Aggregation and Multiplicity Reduction}
When items are sorted by decreasing efficiency with ties broken by larger weights, identical items appear consecutively. Such equal items can therefore be aggregated into a single item by summing their availabilities. To avoid unnecessary overhead, particularly on easy instances, we do not fully sort the items at the beginning. On harder instances, however, the time spent by the algorithm after enumerating a given core \( [l_j, r_j] \) may become much larger than the time required for full sorting. In such cases, full sorting is performed on the ranges \( [1, l_j - 1] \) and \( [r_j + 1, n] \) after the first call of the enhanced divisibility bound. At this stage, we consider applying item aggregation to these intervals and the multiplicity reduction presented in Theorem~\ref{th:mr}.

\begin{theorem}\label{th:mr}
Consider items \(i\) and \(i'\) such that:
\[
p_{i'} = 2p_i,\qquad w_{i'} = 2w_i,
\]
and let \(d_i,d_{i'}\ge 1\) be their availabilities.
Let \(\tilde I\) be the instance obtained from the original instance \(I\) by removing item \(i'\) and replacing its availability by two additional copies of item \(i\), i.e.
\[
\tilde I = I \setminus \{i'\}
\quad\text{and for each } j \in \tilde I,\quad
d_j^{\tilde I} =
\begin{cases}
d_j^{I}, & j \neq i,\\
d_i^{I} + 2 d_{i'}^{I}, & \text{otherwise}.
\end{cases}
\]
Every solution for \(I\) can be transformed into a solution for \(\tilde I\) with the same objective value, and vice versa, i.e., instances $I$ and $\tilde{I}$ are equivalent.
\end{theorem}

Let's suppose items are sorted by decreasing efficiency, with ties broken by larger weight. For items $i$ and $i'$ satisfying Theorem~\ref{th:mr}, we refer to $i$ as the half multiplier of $i'$. Under this ordering, the half multiplier of an item $i$ can only appear to its right, and the half multiplier of item $i+1$ can only appear to the right of the half multiplier of item $i$ (assuming identical items are grouped). This monotonicity property allows the reduction described in Theorem~\ref{th:mr} to be implemented in linear time on each interval $[1, l-1]$ and $[r+1, n]$ using a two-pointer technique. Specifically, we scan the items from left to right using an index $i'$ and maintain a second index $i$, initially set to $i'+1$, which represents the leftmost candidate for being the half multiplier of item $i'$. While $i$ is a valid index for the interval, we compare item $i'$ with an artificial item $j$, defined by $p_j = 2 p_{i}$ and $w_j = 2 w_{i}$, using the sorting comparator. If $j$ is ordered before item $i'$, then the half multiplier of item $i'$ must be to the right of index $i$ (if it exists), so $i$ is incremented. Otherwise, the scan terminates, and index $i$ points to the half multiplier of item $i'$ if it belongs to the item set $I$. Each time a pair of items satisfying Theorem~\ref{th:mr} is found, the reduction is applied. Since both indices $i$ and $i'$ move monotonically to the right, the reduction runs in linear time over each interval.

Once an optimal solution to the reduced instance \(\tilde I\) is found, it can be converted back to a solution for the original instance in linear time using reverse processing together with a copy of the original items $I$.

The preliminary experiments showed that the item aggregation and the multiplicity reduction techniques can drastically reduce the computing time to solve the hardest KP instances, as the number of distinct items in the reduced instance $\tilde I$ may be much smaller than in the original one. By breaking symmetries, they also strengthen the reduction based on weak upper bounds. For example, for a given instance $I$, consider two items $i, i' \ge b$ with $d_i, d_{i'} \geq 3$ and satisfying the conditions of Theorem \ref{th:mr}, and suppose that the weak upper bounds result in $u_i = 2$ and $u_{i'} = 1$. In this case, we can include either two copies of $i$ or one copy of $i'$, but no improved solution combines them. On the other hand, in the instance $\tilde I$, we obtain $u_i = 2$, which is equivalent to fixing all copies of $i'$. Moreover, after applying reduction, if \(\tilde d_i > \lfloor W / w_i \rfloor\), then \(\tilde d_i - \lfloor W / w_i \rfloor\) copies can be fixed.

\begin{remark}
It may appear counterintuitive to apply a reduction that increases the total number of items in the instance. However, since the algorithm proceeds by iterating over a binary decomposition of the item availabilities, even if no items are fixed by bounding techniques after the reduction, the total number of iterations does not increase.
\end{remark}

\begin{corollary}\label{cor:mr-generalized}
The multiplicity reduction can be naturally generalized to arbitrary integer multiplicities.
\end{corollary}

Notice that identifying beneficial values of \(k\) and performing the reverse mapping from the reduced optimal solution to the original instance when multiple values of \(k\) are used are not straightforward. This is why we restrict the reduction to the case \(k = 2\).

After all, the resulting instance after item aggregation and multiplicity reduction is referred to as \(I\) for simplicity.

\subsection{Fixing-by-Dominance}~\label{sec::dominance}
After processing item \(A[j]\), our proposed algorithm keeps in \(S_j\) only the states that are non-dominated and not pruned by the state-based LP bound. Lemma~\ref{lm::first_iter} and Theorem~\ref{th::k_iter} introduce fixing-by-dominance criteria to skip some iterations for a given item \(A[j]\), while Lemmas \ref{lm::dominance_right}, \ref{lm::dominance_left}, and \ref{lm::dominance_right_2} allow us to skip processing the item \(A[j]\) if some conditions hold.

\begin{lemma}\label{lm::first_iter}
Let \(i=A[j]\). If in the first iteration for item \(i\) every extension of \(S_{j-1}\) obtained from adding (removing) one copy of \(i\) produces only states that are either dominated or pruned (i.e., no new state is generated), then all subsequent iterations for \(i\) do not generate new states.
\end{lemma}

\begin{theorem}\label{th::k_iter}
    If the $k$-th iteration for $i = A[j]$ does not generate any new state, then every subsequent iteration for $i$ does not generate any new state.
\end{theorem}

While the fixing-by-dominance criteria in Lemma~\ref{lm::first_iter} and Theorem~\ref{th::k_iter} allow us to skip some iterations of item $i = A[j]$, we are also interested in identifying items that either dominate $i$ or are dominated by $i$, so that they can be completely skipped. Lemmas~\ref{lm::dominance_right}, \ref{lm::dominance_left}, and \ref{lm::dominance_right_2} introduce additional fixing-by-dominance criteria. The items that dominate or are dominated by $i$ are obtained by a linear search, and for this reason, our algorithm applies these criteria only when $|S_j| \ge n$. Furthermore, Lemma~\ref{lm::dominance_right_2} requires a linear search over $S_j$ to verify its conditions, and therefore this search is performed only if at least two items dominated by $i$ are found.

\begin{lemma}\label{lm::dominance_right}
    Let \(i = A[j] \ge b\). If the extension of \(S_{j-1}\) by one copy of \(i\) does not generate any new state, then for any item \(i' > i\) such that
    \[
    w_{i'} \ge w_i \quad \text{and} \quad p_{i'} \le p_i \left\lfloor \frac{w_{i'}}{w_i} \right\rfloor,
    \]
    no improved solution can be obtained from including one or more copies of \(i'\). Consequently, the index \(j'\) in the DP recursion that corresponds to \(i'\) can be skipped.
\end{lemma}

\begin{lemma}\label{lm::dominance_left}
    Let \(i = A[j] < b\). If the extension of \(S_{j-1}\) by one copy of \(i\) does not generate any new state, then for any item \(i' < i\) such that
    \[
    w_{i'} \le w_i \quad \text{and} \quad p_{i'} \ge p_i \left \lceil \frac{w_{i'}}{w_i} \right\rceil = p_i,
    \]
    no improved solution can be obtained from removing one or more copies of \(i'\). Consequently, the index \(j'\) in the DP recursion that corresponds to \(i'\) can be skipped.
\end{lemma}

\begin{lemma}\label{lm::dominance_right_2}
Let \( i = A[j] \ge b \), and assume that in the last iteration for item \( i \) some new states are generated. Let \( s \in S_j \) be the state with the minimum weight \( s_w \) such that an additional copy of \( i \) (if such a copy exists, which is not the case since we already performed the last iteration of $i$) could generate a new state by extending \( s \). If the state \( s \) exists, then an item \( i' > i \) with \( w_{i'} \ge w_i \) and \( p_{i'} \le p_i \left\lfloor \frac{w_{i'}}{w_i} \right\rfloor \) can generate a new state only if \( s_w + w_{i'} - \left( \left\lfloor \frac{w_{i'}}{w_i} \right\rfloor - 1 \right) w_i \le W_{\max} \).
\end{lemma}

\subsection{Heuristics}\label{sec::heuristic}

Designing efficient heuristics for the knapsack problem is challenging, as state-of-the-art exact algorithms are already extremely fast. In practice, heuristics with time complexity above $O(n \log n)$ are rarely competitive, and enumerating DP states may be faster.

We employ three initial heuristics and five primal heuristics during DP enumeration. The \emph{initial heuristics} derive solutions from the extended break solution by: (i) selecting a right item $i$ such that $e \leq d_i$ copies fit; (ii) swapping the least efficient item in the extended break solution with $e \leq d_i$ copies of a right item $i$; and (iii) removing $e \leq d_i$ copies of a left item $i$ to insert one additional copy of the break item~$b$. These heuristics follow the ideas of~\citet{Martello_1999}, run in time complexity $O(n)$, and are used to generate the initial solution (line~2 of Algorithm~\ref{main_algorithm}).

The \emph{Pairing Heuristic} (PH) of~\citet{Martello_1999} is the first primal heuristic. For each item $i$ outside the core $[l_j, r_j]$, it identifies the best state $s \in S_j$ such that $s_w + \overline{w}_i \le W$, where the best state maximizes $s_w$. Since $S_j$ is sorted by weight, this state is found via binary search. We also consider a variant that evaluates every pair $(i_l, i_r)$ of left and right unfixed items outside the core, referred to as the \emph{Two Pairing Heuristic} (TPH).

The third primal heuristic, the \emph{Subset Pairing Heuristic} (SSPH), enumerates all subsets of $k$ selected items in $O(k 2^k)$ time. Given a parameter $\alpha$, we choose the largest $k$ such that $\alpha k 2^k \le |S_j|$, and randomly select $k$ unfixed items outside the core (or all of them if fewer are available). For each non-singleton subset, we perform a pairing step with $S_j$. The resulting complexity is $O(\alpha k 2^k \log |S_j|)$. In our implementation, $\alpha = 20$, yielding overall complexity $O(|S_j| \log |S_j|)$ with a small constant.

The above primal heuristics define the \emph{heavy primal} (HP) heuristic and are triggered so as to limit overhead. PH runs in $O(n \log |S_j|)$ time and is invoked when the number of enumerated states reaches a threshold $T_{\text{HP}}$, initially set to $10n$ and doubled after each call. TPH, with complexity $O(|P| \log |S_j|)$, where $|P|$ denotes the number of pairs $(i_l, i_r)$ considered, and SSPH are executed together with PH when the number of pairs $(i_l, i_r)$ is smaller than $|S_j|/\log |S_j|$ and when $|S_j| \ge n^2$, respectively.

The fourth primal heuristic, the \emph{Sampling Pairing Heuristic} (SPH), mitigates the overhead of executing PH too early while avoiding excessive delay before reaching $T_{\text{HP}}$. SPH applies PH to a subset of $O(\beta)$ items, for a parameter $\beta$. We set $\beta = \lceil |S_j|/50 \rceil$, ensuring $O(|S_j|)$ candidates with a small constant. Let $\gamma = \lfloor n/\beta \rfloor$, $\gamma_{\text{left}} = \min(l_j, \gamma)$, and $\gamma_{\text{right}} = \min(n - r_j, \gamma)$. If $\gamma_{\text{left}}$ or $\gamma_{\text{right}}$ is too small, either the corresponding interval is short, or PH has already been (or will soon be) executed. We therefore use a minimum size parameter $\kappa = 3$. If $\gamma_{\text{left}} \ge \kappa$, we partition the left interval into blocks of size $\gamma_{\text{left}}$ and select one random item per block (ignoring items with $u_i = 0$). The same is done on the right with $\gamma_{\text{right}}$. SPH runs in $O(|S_j| \log |S_j|)$ time with a very small constant; in benchmark instances, it is significantly faster than extending the next DP state, which costs $O(|S_j|)$ and is less cache-friendly. Hence, applying SPH after each item evaluation introduces negligible overhead.

The fifth primal heuristic, the \emph{Greedy Completion Heuristic} (GCH), starts from the incumbent solution and greedily fills the remaining capacity using items in $[r, n]$, where $r$ is the first item after the rightmost selected item in the incumbent. It runs in $O(n)$ time and is invoked after each new incumbent is found, as well as initially using the solution from line~2 of Algorithm~\ref{main_algorithm}.

\section{Computational Experiments}~\label{sec::experiments}
This section presents computational experiments comparing our proposed algorithm, \texttt{RECORD}, with \texttt{COMBO} and \texttt{BOUKNAP}. The implementations of \texttt{COMBO} and \texttt{BOUKNAP} were obtained from Pisinger’s website (\url{https://hjemmesider.diku.dk/~pisinger/codes.html}), using the revised versions from 2002 and 1999, respectively. All experiments were conducted on a machine running Ubuntu 22.04.1 LTS (64-bit), using C++14 with GCC 11.4.0 compiled with flags \texttt{-O3} and \texttt{-march=native}. The hardware consists of an Intel\textsuperscript{\textregistered}\ Xeon\textsuperscript{\textregistered}\ CPU E5-2630 v4 @ 2.20~GHz with 64~GB of RAM. The \texttt{RECORD} source code and experimental results are available at \url{https://gitlab.com/renanfernandofranco/record}.

\texttt{COMBO} pre-allocates all memory for the DP states in advance using a fixed parameter, \texttt{MAX STATES}, which we set to $2^{31}\approx 2 \cdot 10^9$ in our experiments, corresponding to approximately 48 GB of RAM\@. We notice that allocating such a large amount of memory does not affect runtime, since modern operating systems reserve it as virtual memory in a single constant-time operation and allocate physical memory only upon access. In contrast, the other algorithms allocate memory dynamically as needed, with no imposed memory limits.

In the literature, two benchmark sets are particularly relevant: the widely used set proposed by \citet{Pisinger_2005} and the more challenging instances introduced by \citet{Jooken_2022}. In addition, \citet{Miles_2021} applied instance space analysis using a diverse collection of features to characterize where the instances from \citet{Pisinger_2005} lie in the resulting feature space, and proposed new instances to improve the benchmark’s coverage. Although this approach is appealing and yields greater diversity with respect to the selected features, its results indicate that the new instances are not harder than the original ones. Moreover, several of the proposed instances are very easy or even trivial, which contributes to broader space coverage but is of limited practical interest. We believe this can be explained by the fact that genuinely hard knapsack instances require complex structural properties (e.g., the instances later proposed by~\citet{Jooken_2022}) that are difficult to identify or classify through a posteriori analyses relying solely on item profits and weights. In light of this, and since~\citet{Miles_2021} only proposed small instances with 100 items, we focus in this section on the other two benchmarks and report results for the Smith–Miles' benchmark only in the appendices.

Another observation is that we adapted \texttt{COMBO} to use 128-bit integers when computing upper bounds, instead of the original double-precision (64-bit) floating-point type. This modification was necessary because we observed incorrect outputs on the benchmark from~\citet{Jooken_2022}. These errors can be attributed to the extensive enumeration of states with large profits and weights, in which the computation of upper bounds involves values with more than 20 significant decimal places and therefore cannot be represented with full numerical precision using the \texttt{double} type. Although the instances from~\citet{Pisinger_2005} involve large profits and weights, we observed no incorrect results. This can be attributed to the very small initial absolute optimality gap in this benchmark, which directly impacts the numerical precision required for upper-bound computations, and to the less intensive state enumeration process required by these instances, which reduces the risk of rounding errors causing the algorithm to prune optimal solutions. \texttt{RECORD} follows the same approach for computing upper bounds. This way, it is possible to ensure numerical safety for both algorithms across all studied instances.

The following sections briefly describe the benchmark instances from the literature and present computational experiments comparing our solver with \texttt{COMBO} and \texttt{BOUKNAP}. A more detailed study assessing the impact of each feature of our algorithm is provided in Appendix~\ref{ap:feature}.

\subsection{Instances}
The Pisinger dataset~\citep{Pisinger_2005} for the KP consists of the following instance classes:
\emph{uncorrelated}, \emph{weakly correlated}, \emph{strongly correlated}, \emph{inverse strongly correlated}, \emph{almost strongly correlated}, \emph{subset sum}, \emph{uncorrelated with similar weights}, \emph{uncorrelated span}, \emph{weakly correlated span}, \emph{strongly correlated span}, \emph{multiple strongly correlated}, \emph{profit ceiling}, and \emph{circle}. The number of items ranges from $20$ to $10^4$, and profits and weights are bounded by $cR$, where $R$ is a range parameter and $c \in \mathbb{Z}_{+}$ is a small class-dependent constant (see \citet{Pisinger_2005}).

The first seven classes are taken from the literature and extended to include large coefficients, which can significantly affect the runtime of KP algorithms. The first six classes have $R \in \{10^3, 10^4, 10^5, 10^6, 10^7\}$, while \emph{uncorrelated with similar weights} uses $R \in \{10^5, 10^7, 10^8\}$. The remaining six classes are newly designed to be challenging even with small coefficients and use $R = 10^3$. For each class, $n$, and $R$, $H = 100$ instances are provided, and the knapsack capacity of the $h$-th instance ($h = 1, \dots, H$) is defined as $W = \frac{h}{H+1} \sum_{j=1}^n w_j$. This results in a total of 31,800 instances.

The Jooken dataset~\citep{Jooken_2022} comprises 3,240 instances defined by the number of items $n \in \{400,600,800,1000,1200\}$, the knapsack capacity $W \in \{10^6,10^8,10^{10}\}$, and additional parameters controlling the instance structure. One such parameter is the number of groups $g$, which determines how items are partitioned during instance generation, with items in the same group having similar profits and weights. In our experiments, we observed that $g$ noticeably influences practical difficulty. Additionally, it is worth noting that, although the Jooken dataset is challenging for specialized knapsack solvers, the numerical issues inherent to such instances make them practically unsolvable by some existing general-purpose mixed-integer programming~(MIP) solvers. For these instances, we observe that it is important to consider higher-precision arithmetic or operate with rational numbers, which is not the case for some existing MIP solvers.

\subsection{Results in the Pisinger's benchmark for KP}~\label{sec::results_kp}
In Table~\ref{tab::pisinger}, we report results for the Pisinger KP benchmark comparing \texttt{COMBO} and \texttt{RECORD}. Results for \texttt{BOUKNAP} are omitted, as it is dominated by \texttt{COMBO}. Computational experiments involving \texttt{BOUKNAP} and BKP instances are presented in the next section.

Within the 1200-second time limit, both \texttt{COMBO} and \texttt{RECORD} solved all instances. Results are grouped by class and number of items, where $\#\text{inst}$ denotes the number of instances per group. Due to space constraints and non-homogeneous grouping, we report geometric averages. Average runtimes and standard deviations are given in \textbf{milliseconds}. The column \emph{Wins} indicates the number of instances in which a solver is faster. The column \emph{Avg ratio} reports the average instance-wise runtime ratio between \texttt{COMBO} and \texttt{RECORD}; values greater than 1.05 (favoring \texttt{RECORD}) are highlighted in bold. We boldface the larger number of wins and any average time that is at most 5\% worse than the best average time in that row. As runtimes are subject to measurement noise, particularly at the millisecond scale, we consider both solvers equivalent on an instance if their runtimes differ by less than 5\%. If an algorithm solves an instance in less than 100 ms, we perform 19 additional runs and report the median over 20 executions.

\begin{table}[!htbp]
\caption{Comparison between COMBO and RECORD on Pisinger classes. \label{tab::pisinger}}
\centering
\begin{subtable}[t]{0.49\textwidth}
\centering
\renewcommand{\arraystretch}{0.92}
\resizebox{\textwidth}{!}{%
\begin{tabular}{ccc|rrrrrr|r}
\toprule
\multicolumn{3}{c}{} & \multicolumn{3}{c}{\textbf{COMBO}} & \multicolumn{3}{c}{\textbf{RECORD}} & \\
\cmidrule(lr){4-6} \cmidrule(lr){7-9}
\textbf{Class} & n & $\#\text{inst}$ & \makecell{Avg\\time\\(ms)} & \makecell{Std\\time\\(ms)} & \makecell{Wins} & \makecell{Avg\\time\\(ms)} & \makecell{Std\\time\\(ms)} & \makecell{Wins} & \makecell{Avg\\ratio} \\
\midrule
\multirow{8}{*}{\makecell{Uncorrelated}} & 50 & 500 & 0.04 & 0.00 & 0 & \textbf{0.01} & 0.00 & \textbf{500} & \textbf{2.59} \\
 & 100 & 500 & 0.04 & 0.00 & 0 & \textbf{0.02} & 0.00 & \textbf{499} & \textbf{2.17} \\
 & 200 & 500 & 0.05 & 0.01 & 0 & \textbf{0.03} & 0.01 & \textbf{500} & \textbf{1.84} \\
 & 500 & 500 & 0.06 & 0.01 & 3 & \textbf{0.05} & 0.01 & \textbf{486} & \textbf{1.36} \\
 & 1000 & 500 & 0.09 & 0.02 & 25 & \textbf{0.08} & 0.03 & \textbf{406} & \textbf{1.18} \\
 & 2000 & 500 & \textbf{0.14} & 0.04 & 125 & \textbf{0.14} & 0.05 & \textbf{249} & 1.05 \\
 & 5000 & 500 & \textbf{0.29} & 0.10 & \textbf{367} & 0.33 & 0.12 & 43 & 0.87 \\
 & 10000 & 500 & \textbf{0.56} & 0.24 & \textbf{363} & 0.64 & 0.28 & 51 & 0.87 \\
\hline
\multirow{8}{*}{\makecell{Weakly\\correlated}} & 50 & 500 & 0.05 & 0.01 & 4 & \textbf{0.02} & 0.01 & \textbf{489} & \textbf{1.94} \\
 & 100 & 500 & 0.06 & 0.02 & 17 & \textbf{0.03} & 0.02 & \textbf{446} & \textbf{1.64} \\
 & 200 & 500 & 0.07 & 0.03 & 59 & \textbf{0.05} & 0.04 & \textbf{398} & \textbf{1.40} \\
 & 500 & 500 & 0.12 & 0.06 & 113 & \textbf{0.10} & 0.08 & \textbf{334} & \textbf{1.15} \\
 & 1000 & 500 & 0.19 & 0.13 & 156 & \textbf{0.18} & 0.15 & \textbf{264} & \textbf{1.07} \\
 & 2000 & 500 & \textbf{0.31} & 0.28 & \textbf{235} & \textbf{0.32} & 0.32 & 178 & 0.98 \\
 & 5000 & 500 & \textbf{0.69} & 0.68 & \textbf{335} & 0.78 & 0.80 & 83 & 0.88 \\
 & 10000 & 500 & \textbf{1.29} & 1.33 & \textbf{356} & 1.49 & 1.50 & 86 & 0.87 \\
\hline
\multirow{8}{*}{\makecell{Strongly\\correlated}} & 50 & 500 & 0.25 & 739.44 & 103 & \textbf{0.08} & 2.22 & \textbf{374} & \textbf{3.27} \\
 & 100 & 500 & 1.15 & 1082.46 & 72 & \textbf{0.19} & 10.49 & \textbf{420} & \textbf{6.23} \\
 & 200 & 500 & 3.83 & 1209.59 & 25 & \textbf{0.33} & 11.25 & \textbf{472} & \textbf{11.54} \\
 & 500 & 500 & 4.70 & 1575.45 & 14 & \textbf{0.47} & 5.09 & \textbf{471} & \textbf{10.10} \\
 & 1000 & 500 & 4.71 & 1895.37 & 82 & \textbf{0.66} & 2.74 & \textbf{387} & \textbf{7.18} \\
 & 2000 & 500 & 3.59 & 932.26 & 52 & \textbf{0.79} & 1.16 & \textbf{409} & \textbf{4.53} \\
 & 5000 & 500 & 3.07 & 488.48 & 33 & \textbf{1.38} & 1.19 & \textbf{413} & \textbf{2.22} \\
 & 10000 & 500 & 4.46 & 4300.39 & 16 & \textbf{2.20} & 1.22 & \textbf{446} & \textbf{2.03} \\
\hline
\multirow{8}{*}{\makecell{Inverse\\strongly\\correlated}} & 50 & 500 & 0.22 & 594.09 & 52 & \textbf{0.06} & 1.80 & \textbf{422} & \textbf{3.48} \\
 & 100 & 500 & 0.80 & 835.21 & 44 & \textbf{0.14} & 7.41 & \textbf{438} & \textbf{5.67} \\
 & 200 & 500 & 2.46 & 1017.35 & 25 & \textbf{0.26} & 6.06 & \textbf{469} & \textbf{9.40} \\
 & 500 & 500 & 2.64 & 1751.48 & 59 & \textbf{0.45} & 4.35 & \textbf{421} & \textbf{5.90} \\
 & 1000 & 500 & 2.69 & 2632.23 & 113 & \textbf{0.60} & 2.93 & \textbf{357} & \textbf{4.47} \\
 & 2000 & 500 & 2.77 & 2698.74 & 98 & \textbf{0.85} & 1.35 & \textbf{360} & \textbf{3.27} \\
 & 5000 & 500 & 2.81 & 5152.88 & 134 & \textbf{1.59} & 1.13 & \textbf{319} & \textbf{1.77} \\
 & 10000 & 500 & 4.10 & 2155.53 & 117 & \textbf{2.66} & 1.26 & \textbf{331} & \textbf{1.54} \\
\hline
\multirow{8}{*}{\makecell{Almost\\strongly\\correlated}} & 50 & 500 & \textbf{0.09} & 0.28 & \textbf{318} & \textbf{0.09} & 0.28 & 171 & \textbf{1.07} \\
 & 100 & 500 & \textbf{0.23} & 0.65 & \textbf{290} & \textbf{0.23} & 0.65 & 180 & 0.98 \\
 & 200 & 500 & \textbf{0.42} & 1.26 & \textbf{282} & \textbf{0.43} & 1.29 & 182 & 0.98 \\
 & 500 & 500 & \textbf{0.55} & 0.67 & \textbf{341} & 0.64 & 0.62 & 119 & 0.86 \\
 & 1000 & 500 & \textbf{0.72} & 0.89 & \textbf{407} & 0.92 & 0.91 & 58 & 0.79 \\
 & 2000 & 500 & \textbf{0.94} & 0.96 & \textbf{455} & 1.31 & 0.95 & 28 & 0.71 \\
 & 5000 & 500 & \textbf{1.60} & 0.85 & \textbf{483} & 2.63 & 1.23 & 13 & 0.61 \\
 & 10000 & 500 & \textbf{2.89} & 1.33 & \textbf{487} & 4.87 & 2.29 & 8 & 0.59 \\
\hline
\multirow{8}{*}{\makecell{Subset\\sum}} & 50 & 500 & 3.50 & 649.18 & 39 & \textbf{0.30} & 20.67 & \textbf{458} & \textbf{11.48} \\
 & 100 & 500 & 3.95 & 823.83 & 38 & \textbf{0.30} & 2.24 & \textbf{460} & \textbf{13.18} \\
 & 200 & 500 & 3.00 & 837.98 & 68 & \textbf{0.35} & 3.55 & \textbf{424} & \textbf{8.62} \\
 & 500 & 500 & 2.26 & 777.30 & 59 & \textbf{0.33} & 3.63 & \textbf{438} & \textbf{6.74} \\
 & 1000 & 500 & 1.62 & 648.31 & 61 & \textbf{0.38} & 2.85 & \textbf{431} & \textbf{4.20} \\
 & 2000 & 500 & 1.20 & 619.38 & 59 & \textbf{0.44} & 1.64 & \textbf{431} & \textbf{2.74} \\
 & 5000 & 500 & 1.06 & 487.31 & \textbf{271} & \textbf{0.83} & 1.69 & 212 & \textbf{1.28} \\
 & 10000 & 500 & \textbf{1.10} & 159.75 & \textbf{339} & 1.28 & 2.08 & 135 & 0.86 \\
\hline
\multirow{8}{*}{\makecell{Uncorrelated\\with\\Similar\\Weights}} & 50 & 300 & 0.04 & 0.01 & 29 & \textbf{0.02} & 0.02 & \textbf{245} & \textbf{1.78} \\
 & 100 & 300 & 0.07 & 0.04 & 40 & \textbf{0.04} & 0.03 & \textbf{246} & \textbf{1.55} \\
 & 200 & 300 & 0.12 & 0.13 & 53 & \textbf{0.08} & 0.06 & \textbf{240} & \textbf{1.50} \\
 & 500 & 300 & 0.37 & 0.85 & 47 & \textbf{0.21} & 0.14 & \textbf{225} & \textbf{1.77} \\
 & 1000 & 300 & 0.65 & 1.38 & 66 & \textbf{0.44} & 0.27 & \textbf{193} & \textbf{1.48} \\
 & 2000 & 300 & \textbf{0.97} & 1.35 & \textbf{217} & 1.03 & 0.57 & 75 & 0.95 \\
 & 5000 & 300 & \textbf{2.14} & 1.28 & \textbf{260} & 2.85 & 0.86 & 37 & 0.75 \\
 & 10000 & 300 & \textbf{3.57} & 0.84 & \textbf{293} & 5.51 & 1.36 & 4 & 0.65 \\
\hline
\bottomrule
\end{tabular}
}
\end{subtable}
\hfill
\begin{subtable}[t]{0.49\textwidth}
\centering
\renewcommand{\arraystretch}{0.92}
\resizebox{\textwidth}{!}{%
\begin{tabular}{ccc|rrrrrr|r}
\toprule
\multicolumn{3}{c}{} & \multicolumn{3}{c}{\textbf{COMBO}} & \multicolumn{3}{c}{\textbf{RECORD}} & \\
\cmidrule(lr){4-6} \cmidrule(lr){7-9}
\textbf{Class} & n & $\#\text{inst}$ & \makecell{Avg\\time\\(ms)} & \makecell{Std\\time\\(ms)} & \makecell{Wins} & \makecell{Avg\\time\\(ms)} & \makecell{Std\\time\\(ms)} & \makecell{Wins} & \makecell{Avg\\ratio} \\
\midrule
\multirow{9}{*}{\makecell{Uncorrelated\\span}} & 20 & 100 & 0.04 & 0.00 & 0 & \textbf{0.02} & 0.01 & \textbf{100} & \textbf{2.34} \\
 & 50 & 100 & 0.05 & 0.01 & 7 & \textbf{0.02} & 0.02 & \textbf{89} & \textbf{2.19} \\
 & 100 & 100 & 0.07 & 0.04 & 12 & \textbf{0.03} & 0.04 & \textbf{84} & \textbf{2.31} \\
 & 200 & 100 & 0.15 & 0.12 & 11 & \textbf{0.05} & 0.08 & \textbf{85} & \textbf{3.01} \\
 & 500 & 100 & 0.63 & 0.88 & 2 & \textbf{0.11} & 0.21 & \textbf{97} & \textbf{5.89} \\
 & 1000 & 100 & 1.68 & 4.55 & 2 & \textbf{0.19} & 0.46 & \textbf{97} & \textbf{8.78} \\
 & 2000 & 100 & 4.88 & 20.63 & 1 & \textbf{0.34} & 1.00 & \textbf{99} & \textbf{14.25} \\
 & 5000 & 100 & 16.87 & 96.96 & 5 & \textbf{0.82} & 2.80 & \textbf{94} & \textbf{20.66} \\
 & 10000 & 100 & 56.08 & 472.76 & 2 & \textbf{1.52} & 5.70 & \textbf{98} & \textbf{36.82} \\
\hline
\multirow{9}{*}{\makecell{Weakly\\correlated\\span}} & 20 & 100 & 0.04 & 0.02 & 3 & \textbf{0.02} & 0.01 & \textbf{94} & \textbf{1.80} \\
 & 50 & 100 & 0.06 & 0.13 & 13 & \textbf{0.04} & 0.06 & \textbf{75} & \textbf{1.51} \\
 & 100 & 100 & 0.10 & 0.07 & 23 & \textbf{0.07} & 0.04 & \textbf{62} & \textbf{1.44} \\
 & 200 & 100 & 0.28 & 0.76 & 12 & \textbf{0.14} & 0.11 & \textbf{78} & \textbf{1.93} \\
 & 500 & 100 & 1.20 & 2.47 & 4 & \textbf{0.34} & 0.22 & \textbf{93} & \textbf{3.50} \\
 & 1000 & 100 & 3.55 & 8.31 & 0 & \textbf{0.61} & 0.49 & \textbf{99} & \textbf{5.82} \\
 & 2000 & 100 & 13.83 & 37.69 & 2 & \textbf{1.43} & 1.04 & \textbf{98} & \textbf{9.69} \\
 & 5000 & 100 & 55.78 & 271.39 & 4 & \textbf{3.14} & 3.27 & \textbf{96} & \textbf{17.78} \\
 & 10000 & 100 & 242.06 & 896.92 & 1 & \textbf{7.91} & 5.80 & \textbf{99} & \textbf{30.61} \\
\hline
\multirow{9}{*}{\makecell{Strongly\\correlated\\span}} & 20 & 100 & 0.04 & 0.01 & 5 & \textbf{0.02} & 0.02 & \textbf{85} & \textbf{1.71} \\
 & 50 & 100 & 0.06 & 0.04 & 8 & \textbf{0.04} & 0.04 & \textbf{78} & \textbf{1.45} \\
 & 100 & 100 & 0.10 & 0.19 & 13 & \textbf{0.07} & 0.10 & \textbf{73} & \textbf{1.48} \\
 & 200 & 100 & 0.24 & 0.48 & 7 & \textbf{0.11} & 0.14 & \textbf{92} & \textbf{2.10} \\
 & 500 & 100 & 1.20 & 3.16 & 2 & \textbf{0.33} & 0.36 & \textbf{97} & \textbf{3.68} \\
 & 1000 & 100 & 3.88 & 10.49 & 0 & \textbf{0.65} & 0.66 & \textbf{99} & \textbf{5.99} \\
 & 2000 & 100 & 16.07 & 74.62 & 0 & \textbf{1.44} & 1.78 & \textbf{99} & \textbf{11.17} \\
 & 5000 & 100 & 80.81 & 328.32 & 1 & \textbf{4.21} & 3.65 & \textbf{99} & \textbf{19.21} \\
 & 10000 & 100 & 363.85 & 2192.45 & 1 & \textbf{9.44} & 11.66 & \textbf{99} & \textbf{38.55} \\
\hline
\multirow{9}{*}{\makecell{Multiple\\strongly\\correlated}} & 20 & 100 & 0.04 & 0.00 & 14 & \textbf{0.02} & 0.01 & \textbf{85} & \textbf{1.70} \\
 & 50 & 100 & 0.06 & 0.06 & 33 & \textbf{0.04} & 0.04 & \textbf{61} & \textbf{1.45} \\
 & 100 & 100 & 0.11 & 0.14 & \textbf{52} & \textbf{0.09} & 0.17 & 47 & \textbf{1.22} \\
 & 200 & 100 & 0.25 & 0.50 & 47 & \textbf{0.22} & 0.71 & \textbf{51} & \textbf{1.15} \\
 & 500 & 100 & 0.55 & 1.24 & 38 & \textbf{0.46} & 1.04 & \textbf{47} & \textbf{1.20} \\
 & 1000 & 100 & 0.91 & 2.75 & 40 & \textbf{0.82} & 1.69 & \textbf{53} & \textbf{1.12} \\
 & 2000 & 100 & 1.98 & 7.86 & 30 & \textbf{1.54} & 3.54 & \textbf{53} & \textbf{1.29} \\
 & 5000 & 100 & 7.66 & 25.82 & 31 & \textbf{4.46} & 6.20 & \textbf{63} & \textbf{1.72} \\
 & 10000 & 100 & 20.22 & 55.00 & 22 & \textbf{8.62} & 8.48 & \textbf{72} & \textbf{2.34} \\
\hline
\multirow{9}{*}{\makecell{Profit\\ceiling}} & 20 & 100 & 0.07 & 0.03 & 14 & \textbf{0.05} & 0.03 & \textbf{74} & \textbf{1.41} \\
 & 50 & 100 & 0.12 & 0.08 & 18 & \textbf{0.07} & 0.09 & \textbf{72} & \textbf{1.74} \\
 & 100 & 100 & 0.18 & 0.26 & 29 & \textbf{0.08} & 0.20 & \textbf{66} & \textbf{2.25} \\
 & 200 & 100 & 0.26 & 0.93 & 28 & \textbf{0.10} & 0.35 & \textbf{67} & \textbf{2.69} \\
 & 500 & 100 & 0.84 & 5.81 & 17 & \textbf{0.28} & 1.04 & \textbf{80} & \textbf{3.03} \\
 & 1000 & 100 & 0.76 & 20.56 & 10 & \textbf{0.26} & 1.56 & \textbf{87} & \textbf{2.98} \\
 & 2000 & 100 & 2.02 & 66.15 & 14 & \textbf{0.50} & 2.23 & \textbf{79} & \textbf{4.03} \\
 & 5000 & 100 & 7.52 & 501.18 & 31 & \textbf{1.37} & 4.36 & \textbf{61} & \textbf{5.50} \\
 & 10000 & 100 & 12.63 & 1961.76 & 35 & \textbf{2.36} & 6.96 & \textbf{59} & \textbf{5.35} \\
\hline
\multirow{9}{*}{\makecell{Circle}} & 20 & 100 & 0.04 & 0.00 & 7 & \textbf{0.02} & 0.01 & \textbf{85} & \textbf{1.79} \\
 & 50 & 100 & 0.05 & 0.02 & \textbf{52} & \textbf{0.04} & 0.04 & 40 & \textbf{1.22} \\
 & 100 & 100 & \textbf{0.10} & 0.10 & \textbf{71} & \textbf{0.10} & 0.14 & 29 & 0.96 \\
 & 200 & 100 & \textbf{0.20} & 0.24 & \textbf{81} & 0.24 & 0.31 & 19 & 0.83 \\
 & 500 & 100 & \textbf{0.68} & 1.19 & \textbf{82} & 0.86 & 1.10 & 14 & 0.79 \\
 & 1000 & 100 & \textbf{1.37} & 3.47 & \textbf{51} & 1.53 & 2.62 & 36 & 0.89 \\
 & 2000 & 100 & \textbf{3.46} & 7.06 & 36 & \textbf{3.39} & 4.19 & \textbf{59} & 1.02 \\
 & 5000 & 100 & 14.73 & 28.00 & 24 & \textbf{9.45} & 8.99 & \textbf{74} & \textbf{1.56} \\
 & 10000 & 100 & 29.28 & 46.13 & 23 & \textbf{16.31} & 13.70 & \textbf{76} & \textbf{1.79} \\
\hline
\textbf{Total} & -- & 31800 & 0.68 & 1167.12 & 9270 & 0.31 & 4.62 & 20601 & 2.20 \\
\hline
\bottomrule
\end{tabular}
}
\end{subtable}
\end{table}

Among the classes where \texttt{RECORD} is not dominant, four deserve particular attention: \emph{uncorrelated}, \emph{weakly correlated}, \emph{almost strongly correlated}, and \emph{uncorrelated with similar weights}. We observed that \texttt{RECORD} suffered from overhead associated with the implementation choices required to incorporate the new techniques and features. In these classes, the computing time of \texttt{RECORD} presented a linear correlation with the number of items. On the other hand, the computing time of \texttt{COMBO} grew by a factor smaller than~2 even when the number of items doubled. A plausible explanation is that \texttt{COMBO} relies on design choices that reduce cache misses, resulting in an apparently sublinear scaling. Although \texttt{RECORD} is slower, such classes are easy for both solvers, with all instances solved in under 6 milliseconds on average. Another class in which \texttt{RECORD} performed slightly worse is the \emph{circle} class, for which the worst performance is observed on instances with between 200 and 1000 items.

For all remaining classes, \texttt{RECORD} clearly outperformed \texttt{COMBO}, with speedups ranging from a few times to dozens of times compared to \texttt{COMBO} runtime. Notably, in the \emph{strongly correlated} and \emph{inverse strongly correlated} classes, the speedups stem from \texttt{RECORD} having superior heuristics and avoiding the solution of the SBKP, which is harder than the original instance. In the \emph{subset sum} class, it also benefits from the proposed heuristics. In the \emph{uncorrelated span}, \emph{weakly correlated span}, and \emph{strongly correlated span} classes, we notice the positive impact of the item aggregation, multiplicity reduction, and fixing-by-dominance, since such instances contain many identical items or items that are integer multiples of each other. Finally, the superior performance in the \emph{profit ceiling} class is mainly due to the new divisibility bound.

Figure~\ref{fig::figure_combo_RECORD} presents two comparison plots. The left plot compares the average runtimes of \texttt{COMBO} and \texttt{RECORD} when the Pisinger instances are grouped by class, number of items, and range parameter. Under this aggregation, \texttt{RECORD} is at most twice as slow, and its average runtime never exceeds 20 ms. In contrast, \texttt{COMBO} is slower in several groups by factors larger than 2 and sometimes greater than 10, with runtimes reaching the order of seconds.

The right plot shows the cumulative percentage of solved instances as a function of runtime, without grouping. For any given time limit, \texttt{RECORD} has solved more instances, with only a few requiring between 100 ms and 1 second. In contrast, \texttt{COMBO} reaches runtimes of up to 100 seconds.

\begin{figure}[hb!t]
    \centering
    \caption{\texttt{COMBO} Time vs Time Ratio \texttt{COMBO}/\texttt{RECORD}}~\label{fig::figure_combo_RECORD}
    \includegraphics[width=0.99\linewidth]{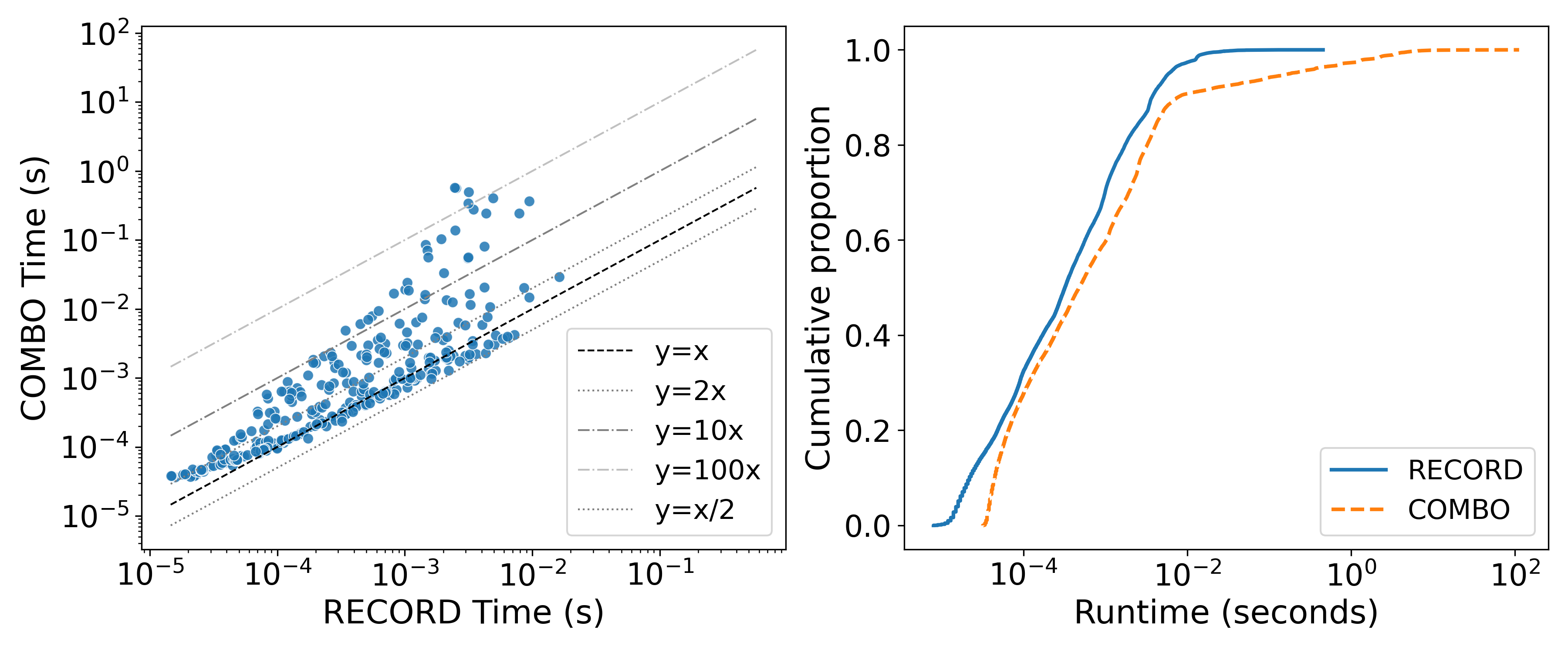}
\end{figure}

\subsection{Results in the Pisinger's benchmark for BKP}~\label{sec::results_bkp}
Seen that \texttt{BOUKNAP} was primarily designed to handle high item availabilities, so we next propose BKP instances that allow the solver to exploit its main features. To this end, we modified David Pisinger’s instance generator, available on his website. Item availabilities were drawn uniformly at random from $[1,20]$, and duplicate items were avoided, except in the three span classes where they were required to reach the desired number of items. For each class, we generated 100 instances for each pair $(n,R) \in {(10^2,10^3),(10^2,10^6),(10^3,10^4),(10^3,10^6),(10^4,10^5),(10^4,10^6)}$. The generation procedure also includes the last six classes with large coefficients.

For \texttt{COMBO}, each BKP instance is converted into a KP instance using the binary decomposition described in Remark~\ref{rmk::decomposition}, except for the classes \emph{Weakly correlated}, \emph{Strongly correlated}, \emph{Inverse strongly correlated}, and \emph{Uncorrelated with similar weights}, where we apply a linear reduction. This choice is motivated by the fact that binary decomposition can adversely affect the surrogate relaxation used by \texttt{COMBO}, leading to significantly worse performance for these classes.

Table~\ref{tab::bounded} compares \texttt{BOUKNAP}, \texttt{COMBO}, and \texttt{RECORD} under a time limit of 20 minutes. Instances are grouped by $n$ and $R$. For each solver, we report the number of instances solved to optimality and the corresponding absolute gap at the time limit. Since \texttt{RECORD} solves all instances, and \texttt{COMBO} solves all instances in the left portion of the table, their optimality and gap columns are omitted there. Runtimes exceeding 100 milliseconds are reported in scientific notation, and average times are highlighted in bold when they are at most 5\% worse than the best average time in the row.

\begin{table}[htbp]
\caption{Comparison between COMBO, BOUKNAP, and RECORD on bounded benchmark classes. \label{tab::bounded}}
\centering
\begin{subtable}[t]{0.46\textwidth}
\centering
\resizebox{\linewidth}{!}{%
\begin{tabular}{l c c c c c c c}
\toprule
\multicolumn{3}{c}{} & \multicolumn{3}{c}{\textbf{BOUKNAP}} & \multicolumn{1}{c}{\textbf{COMBO}} & \multicolumn{1}{c}{\textbf{RECORD}} \\
\cmidrule(lr){4-6} \cmidrule(lr){7-7} \cmidrule(lr){8-8}
Class & $N$ & $R$ & \#OPT & \makecell{Abs\\Gap} & \makecell{Avg\\Time \\(ms)} & \makecell{Avg\\Time \\(ms)} & \makecell{Avg\\Time \\(ms)} \\
\midrule
\multirow{6}{*}{\makecell{Uncorrelated}} & $10^{2}$ & $10^{3}$ & 100 & 0 & \textbf{0.03} & 0.06 & \textbf{0.03} \\
 & $10^{2}$ & $10^{6}$ & 100 & 0 & 0.04 & 0.06 & \textbf{0.03} \\
 & $10^{3}$ & $10^{4}$ & 100 & 0 & 0.26 & 0.24 & \textbf{0.15} \\
 & $10^{3}$ & $10^{6}$ & 100 & 0 & 0.25 & 0.24 & \textbf{0.15} \\
 & $10^{4}$ & $10^{5}$ & 100 & 0 & 2.54 & 2.62 & \textbf{1.47} \\
 & $10^{4}$ & $10^{6}$ & 100 & 0 & 2.65 & 2.72 & \textbf{1.42} \\
\hline
\multirow{6}{*}{\makecell{Weakly\\correlated}} & $10^{2}$ & $10^{3}$ & 100 & 0 & 0.11 & 0.11 & \textbf{0.09} \\
 & $10^{2}$ & $10^{6}$ & 100 & 0 & 0.13 & 0.12 & \textbf{0.09} \\
 & $10^{3}$ & $10^{4}$ & 100 & 0 & 1.10 & 0.72 & \textbf{0.63} \\
 & $10^{3}$ & $10^{6}$ & 100 & 0 & 1.42 & 0.84 & \textbf{0.70} \\
 & $10^{4}$ & $10^{5}$ & 100 & 0 & 12.81 & 8.37 & \textbf{6.25} \\
 & $10^{4}$ & $10^{6}$ & 100 & 0 & 15.81 & 9.86 & \textbf{7.31} \\
\hline
\multirow{6}{*}{\makecell{Strongly\\correlated}} & $10^{2}$ & $10^{3}$ & 100 & 0 & 1.19 & 0.48 & \textbf{0.18} \\
 & $10^{2}$ & $10^{6}$ & 100 & 0 & 1.3e2 & 66.02 & \textbf{2.46} \\
 & $10^{3}$ & $10^{4}$ & 100 & 0 & 1.9e2 & 2.07 & \textbf{0.85} \\
 & $10^{3}$ & $10^{6}$ & 100 & 0 & 3.6e4 & 25.32 & \textbf{11.10} \\
 & $10^{4}$ & $10^{5}$ & 100 & 0 & 2.9e4 & 18.23 & \textbf{3.61} \\
 & $10^{4}$ & $10^{6}$ & 56 & 3.3e4 & 3.7e5 & 23.02 & \textbf{10.56} \\
\hline
\multirow{6}{*}{\makecell{Inverse\\strongly\\correlated}} & $10^{2}$ & $10^{3}$ & 100 & 0 & 1.17 & 0.35 & \textbf{0.15} \\
 & $10^{2}$ & $10^{6}$ & 100 & 0 & 1.3e2 & 25.62 & \textbf{1.44} \\
 & $10^{3}$ & $10^{4}$ & 100 & 0 & 1.8e2 & 2.50 & \textbf{0.81} \\
 & $10^{3}$ & $10^{6}$ & 100 & 0 & 2.3e4 & 14.98 & \textbf{5.95} \\
 & $10^{4}$ & $10^{5}$ & 100 & 0 & 3.5e4 & 20.35 & \textbf{3.85} \\
 & $10^{4}$ & $10^{6}$ & 62 & 2.5e4 & 2.2e5 & 22.97 & \textbf{9.33} \\
\hline
\multirow{6}{*}{\makecell{Almost\\strongly\\correlated}} & $10^{2}$ & $10^{3}$ & 100 & 0 & 0.52 & 0.50 & \textbf{0.24} \\
 & $10^{2}$ & $10^{6}$ & 100 & 0 & 1.00 & 1.07 & \textbf{0.40} \\
 & $10^{3}$ & $10^{4}$ & 100 & 0 & 29.43 & 2.58 & \textbf{1.12} \\
 & $10^{3}$ & $10^{6}$ & 100 & 0 & 39.32 & 3.80 & \textbf{1.26} \\
 & $10^{4}$ & $10^{5}$ & 100 & 0 & 2.1e2 & 68.68 & \textbf{6.87} \\
 & $10^{4}$ & $10^{6}$ & 100 & 0 & 6.2e2 & 1.2e2 & \textbf{9.00} \\
\hline
\multirow{6}{*}{\makecell{Subset\\sum}} & $10^{2}$ & $10^{3}$ & 100 & 0 & 0.66 & 0.21 & \textbf{0.10} \\
 & $10^{2}$ & $10^{6}$ & 100 & 0 & 8.3e2 & 45.50 & \textbf{9.11} \\
 & $10^{3}$ & $10^{4}$ & 100 & 0 & 5.44 & 0.76 & \textbf{0.30} \\
 & $10^{3}$ & $10^{6}$ & 100 & 0 & 1.1e3 & \textbf{2.52} & 5.31 \\
 & $10^{4}$ & $10^{5}$ & 100 & 0 & 84.31 & 3.41 & \textbf{1.76} \\
 & $10^{4}$ & $10^{6}$ & 100 & 0 & 9.4e2 & \textbf{4.20} & 7.86 \\
\hline
\multirow{6}{*}{\makecell{Uncorrelated\\with\\Similar\\Weights}} & $10^{2}$ & $10^{3}$ & 100 & 0 & 0.09 & 0.29 & \textbf{0.04} \\
 & $10^{2}$ & $10^{6}$ & 100 & 0 & 0.14 & 0.32 & \textbf{0.05} \\
 & $10^{3}$ & $10^{4}$ & 100 & 0 & 6.03 & 1.78 & \textbf{0.42} \\
 & $10^{3}$ & $10^{6}$ & 100 & 0 & 6.76 & 2.04 & \textbf{0.47} \\
 & $10^{4}$ & $10^{5}$ & 100 & 0 & 6.9e2 & 75.61 & \textbf{8.25} \\
 & $10^{4}$ & $10^{6}$ & 100 & 0 & 6.6e2 & 93.51 & \textbf{20.42} \\
\hline
\bottomrule
\end{tabular}
}
\end{subtable}%
\hfill%
\begin{subtable}[t]{0.54\textwidth}
\centering
\resizebox{\linewidth}{!}{%
\begin{tabular}{l c c c c c c c c c}
\toprule
\multicolumn{3}{c}{} & \multicolumn{3}{c}{\textbf{BOUKNAP}} & \multicolumn{3}{c}{\textbf{COMBO}} & \multicolumn{1}{c}{\textbf{RECORD}} \\
\cmidrule(lr){4-6} \cmidrule(lr){7-9} \cmidrule(lr){10-10}
Class & $N$ & $R$ & \#OPT & \makecell{Abs\\Gap} & \makecell{Avg\\Time \\(ms)} & \#OPT & \makecell{Abs\\Gap} & \makecell{Avg\\Time \\(ms)} & \makecell{Avg\\Time \\(ms)} \\
\midrule
\multirow{6}{*}{\makecell{Uncorrelated\\span}} & $10^{2}$ & $10^{3}$ & 100 & 0 & 0.81 & 100 & 0 & 0.59 & \textbf{0.03} \\
 & $10^{2}$ & $10^{6}$ & 100 & 0 & 0.79 & 100 & 0 & 0.55 & \textbf{0.04} \\
 & $10^{3}$ & $10^{4}$ & 100 & 0 & 89.28 & 100 & 0 & 15.46 & \textbf{0.20} \\
 & $10^{3}$ & $10^{6}$ & 100 & 0 & 1.1e2 & 100 & 0 & 24.96 & \textbf{0.31} \\
 & $10^{4}$ & $10^{5}$ & 100 & 0 & 9.7e3 & 100 & 0 & 9.0e2 & \textbf{1.85} \\
 & $10^{4}$ & $10^{6}$ & 100 & 0 & 8.0e3 & 100 & 0 & 9.0e2 & \textbf{1.75} \\
\hline
\multirow{6}{*}{\makecell{Weakly\\correlated\\span}} & $10^{2}$ & $10^{3}$ & 100 & 0 & 1.07 & 100 & 0 & 1.03 & \textbf{0.13} \\
 & $10^{2}$ & $10^{6}$ & 100 & 0 & 1.19 & 100 & 0 & 1.10 & \textbf{0.17} \\
 & $10^{3}$ & $10^{4}$ & 100 & 0 & 1.3e2 & 100 & 0 & 77.97 & \textbf{1.38} \\
 & $10^{3}$ & $10^{6}$ & 100 & 0 & 1.2e2 & 100 & 0 & 67.38 & \textbf{1.14} \\
 & $10^{4}$ & $10^{5}$ & 99 & 58.14 & 1.5e4 & 100 & 0 & 5.6e3 & \textbf{17.30} \\
 & $10^{4}$ & $10^{6}$ & 100 & 0 & 1.3e4 & 100 & 0 & 4.3e3 & \textbf{11.50} \\
\hline
\multirow{6}{*}{\makecell{Strongly\\correlated\\span}} & $10^{2}$ & $10^{3}$ & 100 & 0 & 1.02 & 100 & 0 & 1.05 & \textbf{0.14} \\
 & $10^{2}$ & $10^{6}$ & 100 & 0 & 1.32 & 100 & 0 & 1.22 & \textbf{0.17} \\
 & $10^{3}$ & $10^{4}$ & 100 & 0 & 1.7e2 & 100 & 0 & 1.1e2 & \textbf{1.72} \\
 & $10^{3}$ & $10^{6}$ & 100 & 0 & 1.5e2 & 100 & 0 & 84.55 & \textbf{1.63} \\
 & $10^{4}$ & $10^{5}$ & 100 & 0 & 1.3e4 & 100 & 0 & 4.9e3 & \textbf{13.93} \\
 & $10^{4}$ & $10^{6}$ & 98 & 2.9e2 & 1.5e4 & 99 & 1.3e2 & 9.0e3 & \textbf{20.21} \\
\hline
\multirow{6}{*}{\makecell{Multiple\\strongly\\correlated}} & $10^{2}$ & $10^{3}$ & 100 & 0 & 0.49 & 100 & 0 & 0.34 & \textbf{0.20} \\
 & $10^{2}$ & $10^{6}$ & 100 & 0 & 10.16 & 100 & 0 & 4.27 & \textbf{2.24} \\
 & $10^{3}$ & $10^{4}$ & 100 & 0 & 59.88 & 100 & 0 & 33.78 & \textbf{11.03} \\
 & $10^{3}$ & $10^{6}$ & 100 & 0 & 1.0e4 & 100 & 0 & 3.9e3 & \textbf{1.0e3} \\
 & $10^{4}$ & $10^{5}$ & 100 & 0 & 1.1e4 & 100 & 0 & 5.2e3 & \textbf{1.1e3} \\
 & $10^{4}$ & $10^{6}$ & 73 & 3.8e4 & 2.3e5 & 92 & 1.3e4 & 1.1e5 & \textbf{1.7e4} \\
\hline
\multirow{6}{*}{\makecell{Profit\\ceiling}} & $10^{2}$ & $10^{3}$ & 100 & 0 & 0.97 & 100 & 0 & 0.75 & \textbf{0.19} \\
 & $10^{2}$ & $10^{6}$ & 100 & 0 & 2.1e3 & 100 & 0 & 1.1e3 & \textbf{1.2e2} \\
 & $10^{3}$ & $10^{4}$ & 100 & 0 & 51.36 & 100 & 0 & 22.44 & \textbf{2.79} \\
 & $10^{3}$ & $10^{6}$ & 82 & 0.29 & 1.5e4 & 85 & 0.25 & 5.5e3 & \textbf{1.9e2} \\
 & $10^{4}$ & $10^{5}$ & 75 & 0.42 & 9.4e3 & 76 & 0.41 & 2.3e3 & \textbf{65.27} \\
 & $10^{4}$ & $10^{6}$ & 73 & 0.49 & 6.9e4 & 77 & 0.37 & 1.2e4 & \textbf{3.8e2} \\
\hline
\multirow{6}{*}{\makecell{Circle}} & $10^{2}$ & $10^{3}$ & 100 & 0 & 0.42 & 100 & 0 & 0.29 & \textbf{0.24} \\
 & $10^{2}$ & $10^{6}$ & 100 & 0 & 3.56 & 100 & 0 & \textbf{1.28} & \textbf{1.27} \\
 & $10^{3}$ & $10^{4}$ & 100 & 0 & 69.41 & 100 & 0 & \textbf{30.28} & \textbf{30.39} \\
 & $10^{3}$ & $10^{6}$ & 100 & 0 & 2.0e3 & 100 & 0 & \textbf{6.5e2} & 9.7e2 \\
 & $10^{4}$ & $10^{5}$ & 100 & 0 & 1.1e4 & 100 & 0 & \textbf{3.8e3} & 4.4e3 \\
 & $10^{4}$ & $10^{6}$ & 84 & 9.2e4 & 1.5e5 & 100 & 0 & \textbf{4.5e4} & 4.8e4 \\
\hline
Total &  &  & 7602 & 2.4e3 & 73.21 & 7729 & 1.7e2 & 14.65 & \textbf{2.41} \\
\bottomrule
\end{tabular}
}
\end{subtable}
\end{table}

The results show that \texttt{COMBO} is generally stronger than \texttt{BOUKNAP} on this BKP benchmark. This is expected in classes where surrogate relaxation plays a major role, since \texttt{BOUKNAP} does not exploit this feature. Surprisingly, except for a few cases, this advantage persists across all classes, indicating that binary decomposition alone is sufficient to make \texttt{COMBO} competitive on BKP benchmarks. Overall, \texttt{COMBO} solves 7,729 out of 7,800 instances, while \texttt{BOUKNAP} solves 7,602.

\texttt{RECORD} delivers the best performance, solving all 7,800 instances and achieving the lowest average runtime in most groups, often by one to three orders of magnitude in hard classes such as \emph{Strongly correlated span} and \emph{Profit ceiling}. The most difficult cases are concentrated in the last six classes as $R$ increases. For instance, in \emph{Profit ceiling} with $n = 10^2$, the average runtime of \texttt{RECORD} increases from 0.19 ms at $R = 10^3$ to $1.2 \times 10^2$ ms at $R = 10^6$, while for $n = 10^4$ both \texttt{BOUKNAP} and \texttt{COMBO} leave many instances unsolved.
\subsection{Results on Jooken Benchmark}
Table~\ref{tab::new_instances} reports results on the dataset from~\citet{Jooken_2022} with a 1200-second time limit per instance. Columns include the number of instances (\#inst), average runtime and standard deviation (seconds), number of wins (at least 5\% faster in obtaining an optimal solution), average absolute gap (avg abs gap), number of timeouts (TL), and the number of timeouts for which the optimal value was known but not recovered within the time limit ($\text{TL}^*$). Column \emph{Ties} reports instances where the performance difference between \texttt{COMBO} and \texttt{RECORD} is below 5\%, and the last column shows the average runtime ratio between \texttt{COMBO} and \texttt{RECORD}. If both solvers reach the time limit on a given instance, it is classified as a tie, regardless of any differences in their gaps. The \texttt{BOUKNAP} solver was excluded from these experiments because only a negligible number of items have non-unit availabilities.

\begin{table}[!htbp]
\caption{Comparison between COMBO and RECORD on Jooken benchmark. \label{tab::new_instances}}
\centering
\renewcommand{\arraystretch}{0.9}
\resizebox{\textwidth}{!}{%
\begin{tabular}{cccc|rrrrrr|rrrrrr|rr}
\toprule
\multicolumn{4}{c}{} & \multicolumn{6}{c}{\textbf{COMBO}} & \multicolumn{6}{c}{\textbf{RECORD}} & \\
\cmidrule(lr){5-10} \cmidrule(lr){11-16}
$n$ & Capacity & $g$ & $\#\text{inst}$ & \makecell{Avg\\time\\(s)} & \makecell{Std\\time\\(s)} & \makecell{Wins} & \makecell{Avg\\abs\\gap} & \makecell{TL} & \makecell{TL$^{*}$} & \makecell{Avg\\time\\(s)} & \makecell{Std\\time\\(s)} & \makecell{Wins} & \makecell{Avg\\abs\\gap} & \makecell{TL} & \makecell{TL$^{*}$} & \makecell{Ties} & \makecell{Avg\\Time\\Ratio} \\
\midrule
\multirow{6}{*}{400} & $10^{6}$ & 2--6 & 108 & 0.039 & 0.385 & 0 & 0 & 0 & 0 & \textbf{0.007} & 0.044 & \textbf{107} & 0 & 0 & 0 & 1 & \textbf{5.64} \\
 & $10^{6}$ & 10--14 & 108 & 0.003 & 0.050 & 5 & 0 & 0 & 0 & \textbf{0.001} & 0.017 & \textbf{101} & 0 & 0 & 0 & 2 & \textbf{2.54} \\
 & $10^{8}$ & 2--6 & 108 & 0.103 & 2.137 & 0 & 0 & 0 & 0 & \textbf{0.018} & 0.325 & \textbf{107} & 0 & 0 & 0 & 1 & \textbf{5.72} \\
 & $10^{8}$ & 10--14 & 108 & 10.277 & 86.595 & 0 & 0 & 0 & 0 & \textbf{0.874} & 8.218 & \textbf{108} & 0 & 0 & 0 & 0 & \textbf{11.76} \\
 & $10^{10}$ & 2--6 & 108 & 0.120 & 2.676 & 0 & 0 & 0 & 0 & \textbf{0.022} & 0.446 & \textbf{107} & 0 & 0 & 0 & 1 & \textbf{5.56} \\
 & $10^{10}$ & 10--14 & 108 & 340.162 & 564.099 & 0 & 5.3e+3 & 30 & 1 & \textbf{38.750} & 302.402 & \textbf{105} & 0 & 2 & 2 & 3 & \textbf{8.78} \\
\hline
\multirow{6}{*}{600} & $10^{6}$ & 2--6 & 108 & 0.094 & 0.844 & 0 & 0 & 0 & 0 & \textbf{0.012} & 0.076 & \textbf{105} & 0 & 0 & 0 & 3 & \textbf{7.54} \\
 & $10^{6}$ & 10--14 & 108 & 0.003 & 0.056 & 7 & 0 & 0 & 0 & \textbf{0.001} & 0.011 & \textbf{97} & 0 & 0 & 0 & 4 & \textbf{2.59} \\
 & $10^{8}$ & 2--6 & 108 & 0.381 & 7.635 & 0 & 0 & 0 & 0 & \textbf{0.041} & 0.882 & \textbf{108} & 0 & 0 & 0 & 0 & \textbf{9.37} \\
 & $10^{8}$ & 10--14 & 108 & 19.180 & 186.517 & 0 & 0 & 0 & 0 & \textbf{2.095} & 18.189 & \textbf{107} & 0 & 0 & 0 & 1 & \textbf{9.15} \\
 & $10^{10}$ & 2--6 & 108 & 0.409 & 9.058 & 0 & 0 & 0 & 0 & \textbf{0.043} & 1.001 & \textbf{108} & 0 & 0 & 0 & 0 & \textbf{9.39} \\
 & $10^{10}$ & 10--14 & 108 & 501.667 & 495.775 & 0 & 3.1e+5 & 65 & 16 & \textbf{116.903} & 438.380 & \textbf{90} & 2.8e+3 & 16 & 1 & 18 & \textbf{4.29} \\
\hline
\multirow{6}{*}{800} & $10^{6}$ & 2--6 & 108 & 0.130 & 0.697 & 0 & 0 & 0 & 0 & \textbf{0.021} & 0.104 & \textbf{108} & 0 & 0 & 0 & 0 & \textbf{6.08} \\
 & $10^{6}$ & 10--14 & 108 & 0.003 & 0.080 & 6 & 0 & 0 & 0 & \textbf{0.001} & 0.028 & \textbf{98} & 0 & 0 & 0 & 4 & \textbf{2.14} \\
 & $10^{8}$ & 2--6 & 108 & 0.499 & 8.235 & 0 & 0 & 0 & 0 & \textbf{0.073} & 1.555 & \textbf{108} & 0 & 0 & 0 & 0 & \textbf{6.80} \\
 & $10^{8}$ & 10--14 & 108 & 8.255 & 105.364 & 1 & 0 & 0 & 0 & \textbf{1.680} & 21.841 & \textbf{105} & 0 & 0 & 0 & 2 & \textbf{4.91} \\
 & $10^{10}$ & 2--6 & 108 & 0.553 & 9.818 & 0 & 0 & 0 & 0 & \textbf{0.079} & 1.946 & \textbf{108} & 0 & 0 & 0 & 0 & \textbf{7.04} \\
 & $10^{10}$ & 10--14 & 108 & 564.362 & 473.364 & 0 & 5.7e+5 & 67 & 2 & \textbf{187.623} & 459.900 & \textbf{87} & 3.2e+3 & 21 & 4 & 21 & \textbf{3.01} \\
\hline
\multirow{6}{*}{1000} & $10^{6}$ & 2--6 & 108 & 0.220 & 1.016 & 2 & 0 & 0 & 0 & \textbf{0.029} & 0.131 & \textbf{106} & 0 & 0 & 0 & 0 & \textbf{7.49} \\
 & $10^{6}$ & 10--14 & 108 & 0.004 & 0.110 & 14 & 0 & 0 & 0 & \textbf{0.002} & 0.037 & \textbf{88} & 0 & 0 & 0 & 6 & \textbf{1.85} \\
 & $10^{8}$ & 2--6 & 108 & 0.926 & 14.298 & 0 & 0 & 0 & 0 & \textbf{0.113} & 2.606 & \textbf{108} & 0 & 0 & 0 & 0 & \textbf{8.18} \\
 & $10^{8}$ & 10--14 & 108 & 14.065 & 132.758 & 0 & 0 & 0 & 0 & \textbf{2.439} & 27.549 & \textbf{106} & 0 & 0 & 0 & 2 & \textbf{5.77} \\
 & $10^{10}$ & 2--6 & 108 & 1.119 & 22.358 & 0 & 0 & 0 & 0 & \textbf{0.139} & 3.660 & \textbf{108} & 0 & 0 & 0 & 0 & \textbf{8.08} \\
 & $10^{10}$ & 10--14 & 108 & 628.109 & 457.368 & 0 & 4.7e+5 & 72 & 6 & \textbf{266.016} & 496.513 & \textbf{70} & 3.3e+3 & 37 & 14 & 38 & \textbf{2.36} \\
\hline
\multirow{6}{*}{1200} & $10^{6}$ & 2--6 & 108 & 0.362 & 1.405 & 0 & 0 & 0 & 0 & \textbf{0.041} & 0.158 & \textbf{107} & 0 & 0 & 0 & 1 & \textbf{8.82} \\
 & $10^{6}$ & 10--14 & 108 & 0.003 & 0.048 & 10 & 0 & 0 & 0 & \textbf{0.001} & 0.019 & \textbf{91} & 0 & 0 & 0 & 7 & \textbf{1.96} \\
 & $10^{8}$ & 2--6 & 108 & 1.472 & 25.571 & 0 & 0 & 0 & 0 & \textbf{0.155} & 4.862 & \textbf{108} & 0 & 0 & 0 & 0 & \textbf{9.48} \\
 & $10^{8}$ & 10--14 & 108 & 8.628 & 131.636 & 3 & 0 & 0 & 0 & \textbf{2.136} & 24.409 & \textbf{102} & 0 & 0 & 0 & 3 & \textbf{4.04} \\
 & $10^{10}$ & 2--6 & 108 & 1.768 & 36.833 & 0 & 0 & 0 & 0 & \textbf{0.191} & 6.435 & \textbf{108} & 0 & 0 & 0 & 0 & \textbf{9.27} \\
 & $10^{10}$ & 10--14 & 108 & 724.194 & 419.011 & 0 & 3.2e+5 & 81 & 11 & \textbf{322.036} & 491.908 & \textbf{66} & 7.7e+3 & 42 & 12 & 42 & \textbf{2.25} \\
\hline
\textbf{Total} & -- & -- & 3240 & 0.932 & 378.167 & 48 & 5.6e+4 & 315 & 36 & \textbf{0.176} & 262.368 & \textbf{3032} & 5.7e+2 & \textbf{118} & 33 & 160 & \textbf{5.29} \\
\bottomrule
\end{tabular}
}
\end{table}

As observed by~\citet{Jooken_2022}, the most challenging instances are those with a knapsack capacity of $10^{10}$. Furthermore, instance difficulty tends to increase with the number of groups $g$; in particular, instances with $10$ or $14$ groups tend to be the hardest. For brevity, we group the instances by the number of items and knapsack capacity, forming two categories: one with $g \in \{2, 6\}$ and another with $g \in \{10, 14\}$.

As observed, \texttt{RECORD} outperforms \texttt{COMBO} in all instances from this benchmark. In particular, \texttt{COMBO} left 315 instances unsolved, whereas \texttt{RECORD} left only 118, and it achieved a time ratio greater than 1 across all categories. As expected, the instances with $W = 10^{10}$ are more difficult. Additionally, our solver shows a greater improvement in time ratio for $g \in \{2, 6\}$, which is reasonable since, with fewer item groups, each group contains more items, increasing the likelihood of applying the fixing-by-dominance results from Section~\ref{sec::dominance}.
\section{Conclusion}\label{sec::conclusion}
We presented \texttt{RECORD}, a new exact dynamic programming solver for the Knapsack Problem and the Bounded Knapsack Problem. It integrates core-based dynamic programming and SR with several acceleration mechanisms, including multiplicity reduction, refined fixing-by-dominance techniques, new primal heuristics, and an enhanced divisibility bound, resulting in a substantially more efficient exact method.

Across the three experimental sets, \texttt{RECORD} delivered the best overall performance while remaining competitive on easy instances. In the Pisinger KP benchmark, both \texttt{RECORD} and \texttt{COMBO} solved all instances within the 1200-second limit, with \texttt{RECORD} showing large gains on the hardest classes. In the generated BKP benchmark, \texttt{RECORD} solved all 7,800 instances, compared with 7,729 for \texttt{COMBO} and 7,602 for \texttt{BOUKNAP}, and achieved the lowest average time in most groups. On the benchmark of~\citet{Jooken_2022}, \texttt{RECORD} reduced the number of unsolved instances from 315 to 118 and improved average runtime in all categories (overall ratio 5.29), indicating robust behavior across markedly different instance structures.

As future work, we plan to extend our methods to other knapsack variants, such as multidimensional or multiple knapsack problems, and to study the integration of \texttt{RECORD} into broader exact and heuristic frameworks in which the BKP arises as a subproblem.

\appendix

\section{Proofs from the Main Text}

\subsection{Proof of Lemma~\ref{lm::trivial_div}}\label{apx:proof:lm-trivial-div}
\begin{proof}
Since each item \(i \in I_\text{left}\) has \(d_i - u_i\) copies in every improved solution, we construct a residual knapsack using the items in \(I_\text{res}\), where each \(i \in I_\text{res}\) has availability \(u_i\), and the residual capacity is \(\tilde{W}\). It follows that the knapsack capacity of this residual problem can be tightened to the greatest integer \(W'\) such that \(W' \leq \tilde{W}\) and \(W'\) is a multiple of \(\gcd(I_\text{res})\). This implies that every improved solution has capacity at most \(\overline{W}\).
\end{proof}

\subsection{Proof of Theorem~\ref{th::div_rule2}}\label{apx:proof:th-div-rule2}
\begin{proof}
Let \(h \in I^{1}_{\text{left}}\) and consider a modified BKP instance obtained by reducing the availability of \(h\) by one.
By hypothesis, every improved solution for the original instance contains at least \(|I^{1}_{\text{left}}| - 1\) unfixed copies of items in \(I^{1}_{\text{left}}\).
Notice that in the modified instance, all unfixed copies of items in \(I^{1}_{\text{left}} \setminus \{h\}\) are in every improved solution and can therefore be fixed.

As a result, the modified instance has a strictly smaller residual set of items with unfixed copies, yielding a reduced set \(I_{\text{res}}\).
By Lemma~\ref{lm::trivial_div}, we compute a tight knapsack capacity \(\overline{W}_h\) for this instance.
If the corresponding FBKP relaxation with capacity \(\overline{W}_h\) has an optimal value strictly smaller than \(z + 1\), then the modified instance admits no improved solution for the original instance. Consequently, if an improved solution exists, all copies of \(h\) must belong to it, and we can set \(u_h = 0\).
\end{proof}

\subsection{Proof of Theorem~\ref{thm:I2-pairwise}}\label{apx:proof:thm-I2-pairwise}
\begin{proof}
Observe that the first inequality follows directly from the condition \(u_i = 1\). However, the second inequality does not necessarily hold when \(d_i = 1\). We now prove that, assuming the second inequality holds, the statement of the theorem follows. Let \(U\) denote the optimal FBKP value and define the single-copy loss
\[
\delta(i) = U - \mathrm{WB}(i,1).
\]
For any \(i \in I^{1}_{\text{left}}\), it follows that
\(\mathrm{WB}(i,1) = U - \delta(i) \ge z + 1\) and
\(\mathrm{WB}(i,2) = U - 2\delta(i) < z + 1\).

Assume by contradiction that there exist distinct items \(i, i' \in I^{1}_{\text{left}}\) such that
\[
U - \delta(i) - \delta(i') \ge z + 1.
\]
Without loss of generality, we assume that \(\delta(i) \le \delta(i')\). This way,
\[
\mathrm{WB}(i,2) = U - 2\delta(i)
\ge U - \delta(i) - \delta(i')
\ge z + 1,
\]
contradicting \(\mathrm{WB}(i,2) < z + 1\). It follows that no such items $i$ and $i'$ exist.
\end{proof}

\subsection{Proof of Theorem~\ref{th:mr}}\label{apx:proof:th-mr}
\begin{proof}
Let \(x=(x_1,\dots,x_n)\) be a feasible solution for \(I\), where each $x_j$ represents the number of copies of $j \in I$ included in the solution. Let \(\tilde x\) for \(\tilde I\) be defined as:
\[
\tilde x_i = x_i + 2x_{i'},\qquad
\tilde x_{i'} = 0,\qquad
\tilde x_j = x_j \quad\text{for } j\notin\{i,i'\}.
\]
Notice that \(\tilde x_i \le d_i + 2d_{i'} = \tilde d_i\), while total weight and profit are preserved because each copy of \(i'\) contributes the same weight and profit as two copies of \(i\). It follows that \(\tilde x\) is feasible for \(\tilde I\).

Conversely, let \(\tilde x\) be a feasible solution for \(\tilde I\). We may build a solution \(x\) for \(I\) as follows: for all \(j\notin\{i,i'\}\) we set \(x_j=\tilde x_j\) and select
\[
x_{i'} \;=\; \max\!\Big\{0,\;\big\lceil(\tilde x_i - d_i)/2\big\rceil\Big\},
\qquad
x_i \;=\; \tilde x_i - 2x_{i'}.
\]
Since \(\tilde x_i \le \tilde d_i = d_i + 2d_{i'}\), we have \(0 \le x_{i'} \le d_{i'}\). By the choice of \(x_{i'}\) we also have \(0 \le x_i \le d_i\) and \(x_i + 2x_{i'} = \tilde x_i\). Therefore, the total weight and profit of \(x\) are equal to those of \(\tilde x\). It follows that \(x\) is feasible for \(I\). Since this mapping in both directions preserves feasibility and objective value, \(I\) and \(\tilde I\) are equivalent.
\end{proof}

\subsection{Proof of Corollary~\ref{cor:mr-generalized}}\label{apx:proof:cor-mr-generalized}
\begin{proof}
Let $k \in \mathbb{Z}_+$ with $k \ge 2$. If there exist items $i$ and $i'$ such that $p_{i'} = k p_i$, $w_{i'} = k w_i$, $d_i \ge k-1$, and $d_{i'} \ge 1,$ then each copy of $i'$ can be replaced by $k$ copies of $i$ without affecting feasibility or the total profit.
\end{proof}

\subsection{Proof of Lemma~\ref{lm::first_iter}}\label{apx:proof:lm-first-iter}
\begin{proof}
Any extension that adds (or removes) \(d \ge 2\) copies of item \(i\) can be decomposed into \(d\) successive single-copy extensions.
Consequently, a new state \(s\) is generated at iteration \(k\), with \(1 < k \le d\), only if it extends a state \(s'\) generated at iteration \(k-1\).
Such a state \(s'\) must correspond to the extension of some \(s'' \in S_{j-1}\) by \(k-1\) copies of item \(i\).
Since no new state is produced in the first iteration for item \(i\), it follows that no subsequent iteration can generate new states.
\end{proof}

\subsection{Proof of Theorem~\ref{th::k_iter}}\label{apx:proof:th-k-iter}
\begin{proof}
Considering the binary decomposition technique, the $k$-th iteration performs extensions using up to $2^{k}$ copies of item $i$, while the next iteration uses at most $2^{k+1}$ copies. Notice that such a value may be smaller if $k+1$ is the last iteration.

After iteration $k$, the algorithm has considered all possible extensions from a state $s \in S_{j-1}$ using between $0$ and $2^{k+1}-1$ copies of item $i$ (possibly fewer if $k$ is the last iteration). Hence, if a new state $s'$ could be generated in some later iteration $k' > k$, that state would necessarily involve at least $2^{k+1}$ copies of $i$.

However, since the $k$-th iteration did not generate any new state, all states corresponding to up to $2^{k+1}-1$ copies of $i$ must be either dominated or pruned by the LP bound. By the single-copy extension argument in Lemma~\ref{lm::first_iter}, any state containing more copies would also be dominated or pruned, implying that $s'$ cannot exist and no subsequent iteration can generate a new state.
\end{proof}

\subsection{Proof of Lemma~\ref{lm::dominance_right}}\label{apx:proof:lm-dominance-right}
\begin{proof}
According to Lemma~\ref{lm::first_iter}, to add \(h = \lfloor \frac{w_{i'}}{w_i} \rfloor \) copies of item \(i\) cannot generate any new state. Since item \(i'\) is not more efficient than \(h\) copies of \(i\) in terms of weight and profit (i.e., \(w_{i'} \ge h\,w_i\) and \(p_{i'} \le h\,p_i\)), every state obtained from adding one or more copies of \(i'\) is either dominated or pruned by the LP bound. Therefore, no improved solution can include any copy of \(i'\), and we can set \(S_{j'} = S_{j'-1}\).
\end{proof}

\subsection{Proof of Lemma~\ref{lm::dominance_left}}\label{apx:proof:lm-dominance-left}
\begin{proof}
According to Lemma~\ref{lm::first_iter}, to remove \(\left\lceil \frac{w_{i'}}{w_i} \right\rceil = 1\) copy of item \(i\) cannot generate any new state. Since item \(i'\) is at least as efficient as item \(i\), every state obtained from removing one or more copies of \(i'\) is either dominated or pruned by the LP bound. Therefore, no improved solution can be obtained from removing a copy of \(i'\), and we can set \(S_{j'} = S_{j'-1}\).
\end{proof}

\subsection{Proof of Lemma~\ref{lm::dominance_right_2}}\label{apx:proof:lm-dominance-right-2}
\begin{proof}
Recall that any new state must satisfy the capacity limit \(W_{\max}\). Let \(h = \lfloor w_{i'}/w_i \rfloor\) and \(s' \in S_j\) be a state whose extension by \(i'\) generates a new state \(s''\). Let \(\{s' = s^0, \ldots, s^h\}\) be the sequence of states obtained from performing up to \(h\) single-copy extensions of \(s'\) using \(i\). Let \(k\) be the largest index such that \(s^k \in S_j\); clearly \(k < h\) since \(s^h\) dominates \(s''\) and $s''$ is a new state. It follows that,
\[
s^k_w = s''_w - w_{i'} + k w_i \le W_{\max} - w_{i'} + k w_i
\quad \Rightarrow \quad
s^k_w + w_{i'} - k w_i \le W_{\max}.
\]
If the condition holds for $k$, then it also holds for $k+1$. Therefore, the weakest condition is when $k = h - 1$. By hypothesis, the state $s$ exists, and thus $s_w \le s_w^{k}$. This results in the following necessary condition:
\[
s_w + w_{i'} - \bigl(\lfloor w_{i'}/w_i \rfloor - 1\bigr) w_i \le W_{\max}.
\]
\end{proof}

\section{Building the permutation $A^*$}~\label{sec::permutation}
One way to identify the break item \( b \) in linear time is to use an adapted version of the median-finding algorithm, as shown by~\citet{Balas_1980}. For improved practical performance, we instead use a variant of the quickselect algorithm, as also adopted by~\citet{Martello_1999}.

We keep two indices, \( l \) and \( r \), delimiting the current search interval \( [l, r] \). Items in the range \( [1, l) \) belong to the break solution, items in the range \( [r+1, n] \) do not belong to it, and items in the range \( [l, r] \) are uncertain. We define \( ws \) as the weight sum of all items (including all their availabilities) in the range \( [1, l) \). While the interval \( [l, r] \) remains large, a pivot item \( p \) in the interval is selected uniformly at random, and the items in \( [l, r] \) are partitioned using Hoare's algorithm~\citep{cormen_2009} with a comparator that prioritizes items with higher efficiency and breaks ties by larger weights. We notice that breaking ties by sending smaller items to the end often improves the performance of our heuristics. Let \( i \) denote the final position of item \( p \).

Hoare’s algorithm performs a two-way partition, resulting in low constant factors and few swaps. Although two-way partitioning may yield unbalanced partitions when equal elements are not handled explicitly, Hoare’s algorithm mitigates this issue by allowing elements equal to the pivot to appear on either side of the partition index \( i \), which typically leads to balanced partitions in practice.

During the partitioning, we keep track of the weight sum of all items in the left partition \( [l, i - 1] \). If the total weight of \( [l, i] \) (i.e., including the pivot) plus \( ws \) does not exceed the capacity \( W \), all items in \( [l, i] \) are added to the break solution, and their weights are added to the partial sum \( ws \). The interval \( [l, i - 1] \) is then stored for later sorting, and the search continues in the right subinterval \( [i + 1, r] \). Otherwise, the break item belongs to \( [l, i - 1] \). The interval \( [i + 1, r] \) is stored for later sorting, and the search continues in \( [l, i - 1] \). If the interval is small (fewer than $10$ elements), we sort it using \emph{insertion sort} and search on it linearly to identify the break item~$b$. Notice that the algorithm above, from the first iteration where \( l = 1 \) and \( r = n \) to the identification of the break item \( b \), runs in time complexity \( O(n) \) on average.

As a result, in addition to the break item \( b \), we obtain the break solution \( (\widehat{p}, \widehat{w}) \) and a set of intervals that still need to be sorted. These intervals can be classified as either left or right relative to the break item. Naturally, the intervals are relatively sorted: the first left interval contains items that precede those in the second left interval, and the same holds for the right intervals. These intervals are sorted when needed.

Suppose that the current core is \( [l_j, r_j] \) and we need to evaluate a left item. We call \emph{lazy sorting algorithm} the procedure for finding the next unfixed left (or right). An immediate candidate for the next unfixed left item is the item at position \( i = l_j - 1 \). The item \( i \) may be in an unsorted interval, which can be detected in time complexity \( O(1) \) by checking whether \( i \) lies in the last left interval \( [l_{\text{last}}, r_{\text{last}}] \). If this is the case, we locate the correct item for position \( i \) by continuing the sorting process with an algorithm similar to the one used to find \( b \). The latter runs in \( O(r_{\text{last}} - l_{\text{last}} + 1) \) and replaces \( [l_{\text{last}}, r_{\text{last}}] \) with smaller left intervals. If the interval is small, we simply sort it directly using insertion sort.

We observe that the above algorithm may sort items that could be fixed by fixing techniques, which leads to unnecessary computational effort. To address this issue, we check whether the item at position~$i$ can be completely fixed by the weaker upper bound presented in Section~\ref{sec::weak_upper}, which has time complexity~$O(1)$. If so, we proceed to position $i-1$, and so on. If we find a position~$i$ where the item cannot be completely fixed and it lies in an unsorted interval~$[l, i]$, we apply the weaker upper bound to all items in $[l, i]$, and split it into $[l, r]$ and $[r+1, i]$. The interval $[r+1, i]$ contains only fully fixed items, so we can perform the lazy sort only on~$[l, r]$, since the candidate item position is now~$r$.

Calling the lazy sorting algorithm for the right item is analogous. Thus, our ordering~$A^*$ is created by selecting unfixed left and right items in an alternating fashion, whenever both sides are not yet fully processed. Preliminary computational experiments showed this algorithm improved the linear relaxation bound more quickly. In contrast, \texttt{COMBO} adopts a different strategy: it also alternates between left and right items, but if the item at position $l-1$ can be fixed, it then evaluates $r+1$ instead of continuing to $l-2$. In other words, \texttt{COMBO} does not enforce the processing of unfixed left and right items in an interleaved manner.

It is important to mention that if our algorithm performs a constant number of calls to the lazy sorting algorithm before proving the optimality of the incumbent solution, the time complexity to obtain~\(A^*\) is \(O(n)\) on average, since only a constant number of quickselect calls are performed; in other words, it is not necessary to sort all items completely.

\section{Recovering the solution}~\label{sec::recover}
Keeping every evaluated state with an explicit pointer to its parent would make solution recovery a trivial process, but it is very memory-consuming. State-of-the-art solvers were designed in an era when memory was scarce, and therefore they discard the previous state set \(S_{j-1}\) after computing \(S_j\). To achieve similarly low memory usage, we also discard previous states in our algorithm, which indeed may require a sophisticated procedure to reconstruct the optimal solution. While this choice necessitates a more elaborate procedure to reconstruct an optimal solution, it is necessary to handle the challenging instances proposed by~\citet{Jooken_2022}.

Each state \(s\) keeps its profit \(s_p\), weight \(s_w\), and a 64-bit mask \(s_m\). At index \(j\) of the DP recursion, we keep a change vector $
V = \bigl((i_0,m_0),(i_1,m_1),\dots,(i_{63},m_{63})\bigr)$, which saves the items and their number of copies used in the last 64 iterations. For each state \(s\), bit \(k\) of \(s_m\) is set if and only if \(s\) is created by extending its ancestor with \(m_k\) copies of item \(i_k\).

When an improved solution is found, we save three elements: the incumbent state \(s^{\text{inc}}\), the current change vector \(V_{\text{inc}}\), and the current core \([l_{\text{inc}},r_{\text{inc}}]\). If the incumbent is obtained from a heuristic, the core is defined as the smallest interval containing \([l_j,r_j]\) and all items used by the solution.

After the DP recursion is solved, we recover parts of the optimal solution implied by \(V_{\text{inc}}\) and \([l_{\text{inc}},r_{\text{inc}}]\) in the following way:
\begin{itemize}
  \item For each \((i_k,m_k)\in V_{\text{inc}}\) with \(i_k < b\), disregard \(m_k\) copies of item \(i_k\) in the solution.
  \item For each \((i_k,m_k)\in V_{\text{inc}}\) with \(i_k \ge b\), fix \(m_k\) copies of item \(i_k\) in the solution.
  \item All copies of items \(i < l_{\text{inc}}\) are fixed in the solution, while all items \(i > r_{\text{inc}}\) are discarded.
\end{itemize}

Once the total value of the fixed items is equal to \(s_p^{\text{inc}}\), the solution is completely recovered. Otherwise, the fixed items define a prefix solution \((p_{\text{pref}},w_{\text{pref}})\), and a reduced knapsack instance is solved to recover the remaining items. The reduced instance has knapsack capacity \( W' = s_w^{\text{inc}} - w_{\text{pref}} \) and an initial upper bound \( s_p^{\text{inc}} - p_{\text{pref}} \), which allows us to terminate as soon as a state representing an optimal solution is found. An additional observation not explored by \texttt{COMBO} is that we may assume an incumbent solution with value \( z = s_p^{\text{inc}} - p_{\text{pref}} - 1 \) is known. This strong bound allows us to prune many states that cannot lead to an optimal solution.

In the preliminary experiments, we noticed that the reduced instances are typically much easier than the original ones. Although in the worst case we may need to solve \(O(n)\) reduced instances, in practice the number is small (usually fewer than four), and the total recovery time rarely exceeds 30\% of the overall computing time.

\section{Surrogate Relaxation with Cardinality Constraints}~\label{ap::sr}
The state-based LP bound usually results in small relative gaps for the BKP, although the absolute gaps may still be large when item profits are high. Moreover, even in instances with large initial gaps, it typically decreases quickly after a few items have been enumerated. There are instances for which the state-based LP bound does not improve quickly, even after enumerating many of the items. An example of an instance where this happens is when every optimal integer solution may contain exactly \( K \) items, while an LP solution \( x^\text{LP} \) may take a much larger value, with \( \sum_{i = 1}^n x^\text{LP}_i = K' \notin \mathbb{Z}_+ \) and \( K < K' < K + 1 \).

Following~\citet{Martello_1997} for the KP, if an oracle can guarantee that at least one optimal solution has between $L$ and $K$ items, then the following cardinality constraints can be added to the BKP\@:
\begin{align}
    \sum_{i=1}^n x_i &\ge L \text{ (minimum cardinality)},~\label{eq:min-card}\\
    \sum_{i=1}^n x_i &\le K \text{ (maximum cardinality)}.~\label{eq:max-card}
\end{align}

Identifying tight values of $L$ and $K$ is not straightforward for general BKP instances. Next, we introduce two bounds proposed by~\citet{Martello_1997} that may help with this objective.
Given an incumbent value $z$, any improved solution must satisfy
\[
    N_{\min} \;=\; \min \left\{ \sum_{i=1}^n x_i \;:\;
    \sum_{i=1}^n p_i x_i \ge z+1,\; x_i \le d_i,\; x_i \in \mathbb{Z}^+, \forall i \in \{1, \ldots, n\} \right\}.
\]
Conversely, every feasible solution must satisfy
\[
    N_{\max} \;=\; \max \left\{ \sum_{i=1}^n x_i \;:\;
    \sum_{i=1}^n w_i x_i \le W,\; x_i \le d_i,\; x_i \in \mathbb{Z}^+, \forall i \in \{1, \ldots, n\} \right\}.
\]
Both \( N_{\min} \) and \( N_{\max} \) can be computed in linear time using an adapted version of the quickselect algorithm. A first option is to set \( L = N_{\min} \) and \( K = N_{\max} \), but additional alternatives are presented at the end of this section.

Incorporating cardinality constraints directly into the DP recurrence is challenging, since they transform the problem into a two- or three-constraint KP\@. One solution found in the literature is to relax them to obtain stronger upper bounds through \emph{Lagrangian relaxation}~\citep{Martello_1997} and \emph{surrogate relaxation}~\citep{Martello_1999}. Since both cardinality constraints are structurally analogous, we focus on the one given in~\eqref{eq:max-card}.
In the Lagrangian relaxation, the maximum-cardinality constraint is dualized and included in the objective function with a multiplier~$\lambda$. This results in the Lagrangian Bounded Knapsack Problem~(LBKP).
In contrast, the SR aggregates both the maximum-cardinality and capacity constraints into a single inequality using a multiplier~$\mu$, resulting in the following Surrogate Bounded Knapsack Problem (SBKP):
\begin{align}
    \text{maximize} \quad & \sum_{i = 1}^{n} p_i x_i &&~\label{eq:sbkp-obj} \\
    \text{subject to} \quad & \sum_{i=1}^n w_i x_i + \mu \sum_{i=1}^n x_i \;\le\; W + K \mu, &&~\label{eq:sbkp-cap} \\
    & x_i \leq d_i, & \text{for all } i \in \{1, \ldots, n\}, &&~\label{eq:sbkp-bounds} \\
    & x_i \in \mathbb{Z}_+, & \text{for all } i \in \{1, \ldots, n\}. &&~\label{eq:sbkp-int}
\end{align}
Determining the multipliers that give the tightest upper bounds for LBKP and SBKP is computationally expensive. However, we can relax the integrality constraints of both problems, resulting in the Lagrangian Fractional Bounded Knapsack Problem~(LFBKP) and Surrogate Fractional Bounded Knapsack Problem~(SFBKP). Since both LFBKP and SFBKP are convex, they can be solved with binary search or gradient procedures.

It is important to note that the optimal multipliers $\lambda^*$ and $\mu^*$ for LFBKP and SFBKP, respectively, yield identical upper bounds. These bounds are stronger than the LP bound only if the FBKP solution~\eqref{eq:optimal_fract} violates their respective cardinality constraints \eqref{eq:min-card} or \eqref{eq:max-card}. Moreover, since benchmark instances in the literature use integer parameters and the optimal multipliers do not need to be integers,~\citet{Martello_1999} commented that the best integer-valued multiplier is often sufficient in practice, conveniently avoiding numerical precision issues associated with floating-point arithmetic. The computational experiments have confirmed this observation: the best integer multipliers $\lambda^*_{\text{int}}$ and $\mu^*_{\text{int}}$ almost always produce the same upper bounds, with only a few exceptions. Moreover, once $\lambda^*_{\text{int}}$ or $\mu^*_{\text{int}}$ is obtained, we can solve LBKP or SBKP using these multipliers and obtain an even tighter bound. As shown by~\citet{Martello_1999}, the resulting BKP is typically easier to solve, and its solution is often feasible, or even optimal, for the original problem. In particular, optimality is guaranteed when the cardinality constraint is satisfied with equality, and even when this does not occur, a feasible solution may still improve the incumbent one.

Although LFBKP and SFBKP are equivalent, no such claim can be made for the corresponding integer problems, LBKP and SBKP, when using the multipliers $\lambda^*_{\text{int}}$ and $\mu^*_{\text{int}}$, respectively. With this in mind, we implemented both approaches and compared them computationally. Our computational experiments indicate that there is no dominance between them, but SBKP provides stronger bounds than LBKP far more frequently than the converse. These findings are consistent with the fact that SR dominates Lagrangian relaxation in the integer case when optimal multipliers are used~\citep{Glover_1975}. That is, this dominance manifests empirically even when using suboptimal multipliers $\lambda^*_{\text{int}}$ and $\mu^*_{\text{int}}$, leading the surrogate approach to exhibit a consistent tendency to perform better.

After that, we decided to keep our algorithm with only the SR, following~\citet{Martello_1999}, with the optimal integer multiplier $\mu^{*}_{\text{int}}$ computed via binary search, as proposed in their work. Let $p_{\max} = \max_i p_i$ and $w_{\max} = \max_i w_i$.
For the maximum-cardinality constraint, we have that
$\mu^{*}_{\text{int}} \in [0,\, p_{\max} w_{\max}]$.
For the minimum-cardinality one, we can represent it in the form of constraint~\eqref{eq:sbkp-cap} with a nonpositive multiplier, which implies
$\mu^{*}_{\text{int}} \in [-w_{\max},\, 0]$.

Let $\Gamma(\mu)$ denote the minimum value of $\sum_{i = 1}^n x_i$ for an optimal SFBKP solution $x$ under multiplier~$\mu$.
To compute it efficiently, we notice that ties between items of equal efficiency can be broken by choosing items with smaller weight values first.
For $\mu \ge 0$,~\citet{Martello_1999} proved that $\Gamma(\mu)$ is monotonically non-increasing as $\mu$ increases.
This property extends naturally to $\mu < 0$.
Moreover, any item with non-positive effective weight $w_i + \mu \le 0$ can always be taken, since its profits are positive and the knapsack capacity is increased by $-(w_i + \mu)$.

We consider the binary search procedure of \citet{Martello_1999} to calculate $\mu^{*}_{\text{int}}$.
For the maximum-cardinality constraint, we initialize $\mu_1 \gets 0$ and
$\mu_2 \gets p_{\max} w_{\max}$, and compute $\mu' = \lfloor (\mu_1 + \mu_2) / 2 \rfloor$. If \(\Gamma(\mu') = K\), we obtain the optimal multiplier.
If \(\Gamma(\mu') > K\), we set \(\mu_2 \gets \mu' - 1\), and if \(\Gamma(\mu') < K\), we set \(\mu_1 \gets \mu' + 1\). The search ends when \([ \mu_1, \mu_2 ]\) becomes empty, and \(\mu^{*}_{\text{int}}\) is the value of \(\mu'\) that provides the tightest bound.

It is important to mention that we observed an issue during the computational experiments in large instances. The product $p_{\max} w_{\max}$ may overflow a 64-bit integer. Moreover, solving SFBKPs using a higher-precision arithmetic (e.g., 128-bit integers) introduces noticeable overhead. \texttt{COMBO} addresses this issue by setting $\mu_2 \gets w_{\max}$, a choice that may lead to weaker upper bounds. We overcome this issue by using a binary search with \emph{expanding bounds}.
We initialize $\mu_1 \gets 0$ and $\mu_2 \gets w_{\max}$, and solve the SFBKP for $\mu_2$.
If $\Gamma(\mu_2) < K$, then $\mu^{*}_{\text{int}} \ge \mu_2$, so we set
$\mu_1 \gets \mu_2 + 1$ and update $\mu_2 \leftarrow 2\mu_2$.
We continue doubling $\mu_2$ while $\Gamma(\mu_2) < K$ and
\[
    \mu_2 \;\le\;
    \frac{2^{63} - 1}{4 \sum_{i=1}^n d_i},
\]
which serves as a safety threshold to avoid overflow.
Once the expansion phase stops, we proceed with the standard binary search on the resulting interval $[\mu_1,\, \mu_2]$. The overall algorithm has time complexity $O(n \log(p_{\max} w_{\max}))$. The optimal value of SFBKP for $\mu^{*}_{\text{int}}$ is denoted $UB^{SF}$ and consists of an upper bound for the original problem. If \(UB^{SF} > z\), then the current incumbent solution cannot be optimal, and we create a candidate SBKP instance with the multiplier \(\mu^{*}_{\text{int}}\), which is solved using our own algorithm without SRs and multiplicity reduction.

We also consider some additional tricks. There are classes of instances, such as \emph{strongly correlated} and \emph{inverse strongly correlated}, for which there exists a constant
$C \in \mathbb{Z}$ such that $p_i = w_i + C$ for all items~$i$.
In these instances, $C$ tends to be exactly the optimal multiplier
$\mu^{*}_{\text{int}}$.
Thus, we simply test whether $C$ is optimal by comparing it with its two neighbors $C-1$ and $C+1$.
If \( C \) is the best among these three values, then the optimal multiplier is obtained in \( O(n) \). In this case, we additionally observed that solving the corresponding SBKP instance is typically unproductive, as it reduces to a subset-sum problem, which is generally significantly harder than the original BKP, and it almost always provides an upper bound that coincides with that of SFBKP. Consequently, we decide not to solve the SBKP for these instances.
\section{Feature Behavior Comparison}~\label{ap:feature}
We evaluate ablated versions of the algorithm, each obtained by disabling one feature or group of features: \emph{completion features} (maximum capacity and minimum profit bounds), \emph{guarded extension}, the new \emph{divisibility bound}, \emph{skip subs iter} (abbrev.\ for skipping subsequent iterations; Theorem~\ref{th::k_iter}), \emph{multiplicity reduction}, \emph{item aggregation with multiplicity reduction}, \emph{fixing-by-dominance} (Lemmas~\ref{lm::dominance_right}--\ref{lm::dominance_right_2}), \emph{item aggregation with multiplicity reduction and fixing-by-dominance}, and the heuristics \emph{TPH}, \emph{SSPH}, \emph{SPH}, and \emph{GCH}.

Since multiplicity reduction depends on item aggregation, it cannot be disabled independently; therefore, we also consider a variant in which both are disabled. Moreover, several features exhibit partial redundancy. To illustrate this, we include a variant where item aggregation, multiplicity reduction, and fixing-by-dominance are all disabled simultaneously.

Experiments are conducted on the Pisinger KP benchmark and on all Jooken instances with $W = 10^6$. Instances are grouped by class and number of items. The complete solver and each variant are executed 20 times, and the median runtime is reported. The table presents the ratio between the runtime of each variant and that of the complete solver. Ratios within 5\% of one are rounded to $1$, and rows consisting entirely of ones are omitted.

Overall, each feature contributes positively to performance. Although fixing-by-dominance alone has a limited impact on the Pisinger instances, disabling it together with item aggregation and multiplicity reduction leads to a substantially larger degradation than disabling only the latter two, revealing partial redundancy. Similar interactions are observed among the heuristics and between the new divisibility bound and other fixing techniques.

\begin{table}[htbp]
\centering
\small
\caption{Time ratio comparison when disabling groups of features from \texttt{RECORD}.}
\label{tab:feature_ablation_time_ratio_all}
\resizebox{0.94\textwidth}{!}{%
\begin{tabular}{r r r r r r r r r r r r r r}
\hline
Class & N & \makecell{completion\\features} & \makecell{guarded\\extension} & \makecell{divisibility\\bounds} & \makecell{skip\\subs\\ iter} & \makecell{multiplicity\\reduction} & \makecell{item agg\\mult red} & \makecell{dominance\\ fixing} & \makecell{item agg\\mult red\\dom fix} & \makecell{TPH} & \makecell{SSPH} & \makecell{SPH} & \makecell{GCH} \\
\hline
\multirow{2}{*}{\makecell{Uncorrelated}} & 50 & 1 & 1.05 & 1 & 1 & 1 & 1 & 1 & 1 & 1 & 1 & 1 & 1 \\
 & 200 & 1 & 1.05 & 1 & 1.06 & 1 & 1 & 1 & 1 & 1 & 1 & 1 & 1 \\
\hline
\noalign{\vskip 1.5pt}
\multirow{6}{*}{\makecell{Strongly\\correlated}} & 100 & 1 & 1 & 1 & 1 & 1 & 1 & 1 & 1 & 1.06 & 1 & 1 & 1 \\
 & 200 & 1 & 1 & 1 & 1 & 1 & 1 & 1 & 1 & 1.05 & 1 & 1.05 & 1 \\
 & 1000 & 1 & 1 & 1 & 1 & 1 & 1 & 1 & 1 & 1 & 1 & 1.18 & 1 \\
 & 2000 & 1 & 1 & 1 & 1 & 1 & 1 & 1 & 1 & 1 & 1 & 1.23 & 1 \\
 & 5000 & 1 & 1 & 1 & 1 & 1 & 1 & 1 & 1 & 1 & 1 & 1.43 & 1 \\
 & 10000 & 1 & 1 & 1 & 1 & 1 & 1 & 1 & 1 & 1 & 1 & 1.44 & 1 \\
\hline
\noalign{\vskip 1.5pt}
\multirow{5}{*}{\makecell{Inverse\\strongly\\correlated}} & 200 & 1 & 1 & 1 & 1 & 1 & 1 & 1 & 1 & 1 & 1 & 1.05 & 1 \\
 & 1000 & 1 & 1 & 1 & 1 & 1 & 1 & 1 & 1 & 1 & 1 & 1.17 & 1 \\
 & 2000 & 1 & 1 & 1 & 1 & 1 & 1 & 1 & 1 & 1 & 1 & 1.20 & 1 \\
 & 5000 & 1 & 1 & 1 & 1 & 1 & 1 & 1 & 1 & 1 & 1 & 1.35 & 1 \\
 & 10000 & 1 & 1 & 1 & 1 & 1 & 1 & 1 & 1 & 1 & 1 & 1.38 & 1 \\
\hline
\noalign{\vskip 1.5pt}
\multirow{8}{*}{\makecell{Subset\\sum}} & 50 & 1.16 & 1 & 1 & 1 & 1 & 1 & 1 & 1 & 1.13 & 1.55 & 1 & 1 \\
 & 100 & 1 & 1 & 1 & 1 & 1 & 1 & 1 & 1 & 1.19 & 1.14 & 1.08 & 1 \\
 & 200 & 1 & 1 & 1 & 1 & 1 & 1 & 1 & 1 & 1 & 1.06 & 1.21 & 1 \\
 & 500 & 1 & 1 & 1 & 1 & 1 & 1 & 1 & 1 & 1 & 1 & 1.23 & 1.06 \\
 & 1000 & 1 & 1 & 1 & 1 & 1 & 1 & 1 & 1 & 1 & 1 & 1.43 & 1 \\
 & 2000 & 1 & 1 & 1 & 1 & 1 & 1 & 1 & 1 & 1 & 1 & 1.38 & 1 \\
 & 5000 & 1 & 1 & 1 & 1 & 1 & 1 & 1 & 1 & 1 & 1 & 1.40 & 1 \\
 & 10000 & 1 & 1 & 1 & 1 & 1 & 1 & 1 & 1 & 1 & 1 & 1.45 & 1 \\
\hline
\noalign{\vskip 1.5pt}
\multirow{9}{*}{\makecell{Uncorrelated\\span}} & 20 & 1 & 1 & 1 & 1.22 & 1 & 1 & 1 & 1 & 1 & 1 & 1 & 1 \\
 & 50 & 1.08 & 1 & 1 & 1.06 & 1 & 1 & 1 & 1 & 1 & 1 & 1 & 1 \\
 & 100 & 1.05 & 1 & 1 & 2.51 & 1 & 1 & 1 & 1.05 & 1 & 1 & 1 & 1 \\
 & 200 & 1.06 & 1 & 1 & 1.28 & 1 & 1.05 & 1 & 1.08 & 1 & 1 & 1 & 1 \\
 & 500 & 1 & 1 & 1 & 1.09 & 1 & 1.12 & 1 & 1.18 & 1 & 1 & 1 & 1 \\
 & 1000 & 1 & 1 & 1 & 2.41 & 1 & 1.33 & 1 & 1.63 & 1 & 1 & 1 & 1 \\
 & 2000 & 1 & 1 & 1 & 1.50 & 1 & 1.27 & 1 & 1.59 & 1 & 1 & 1 & 1 \\
 & 5000 & 1 & 1 & 1 & 1.10 & 1 & 1.45 & 1 & 1.79 & 1 & 1 & 1 & 1 \\
 & 10000 & 1.28 & 1.29 & 1.28 & 3.06 & 1.30 & 2.24 & 1.25 & 2.66 & 1.27 & 1.27 & 1.23 & 1.26 \\
\hline
\noalign{\vskip 1.5pt}
\multirow{9}{*}{\makecell{Weakly\\correlated\\span}} & 20 & 1.09 & 1 & 1 & 1 & 1 & 1 & 1 & 1 & 1 & 1 & 1 & 1 \\
 & 50 & 1.14 & 1 & 1 & 1 & 1 & 1 & 1 & 1 & 1 & 1 & 1 & 1 \\
 & 100 & 1.13 & 1 & 1 & 1 & 1 & 1.07 & 1 & 1.09 & 1 & 1 & 1 & 1 \\
 & 200 & 1.12 & 1 & 1 & 1.08 & 1 & 1.22 & 1 & 1.32 & 1 & 1 & 1 & 1 \\
 & 500 & 1.11 & 1 & 1 & 1.08 & 1.05 & 1.59 & 1 & 2.11 & 1 & 1 & 1 & 1 \\
 & 1000 & 1.06 & 1 & 1 & 1.09 & 1 & 2.86 & 1 & 3.61 & 1 & 1 & 1 & 1 \\
 & 2000 & 1.07 & 1 & 1 & 1.11 & 1.05 & 2.43 & 1 & 5.19 & 1.05 & 1 & 1 & 1 \\
 & 5000 & 1.06 & 1 & 1 & 1.11 & 1.06 & 4.42 & 1 & 9.74 & 1 & 1 & 1 & 1 \\
 & 10000 & 1 & 1 & 1 & 1.15 & 1.06 & 6.77 & 1 & 15.82 & 1 & 1 & 1 & 1 \\
\hline
\noalign{\vskip 1.5pt}
\multirow{9}{*}{\makecell{Strongly\\correlated\\span}} & 20 & 1.18 & 1 & 1 & 1 & 1 & 1 & 1 & 1.07 & 1 & 1 & 1 & 1 \\
 & 50 & 1.17 & 1 & 1 & 1 & 1 & 1 & 1 & 1.07 & 0.94 & 1 & 1 & 1 \\
 & 100 & 1.17 & 1 & 1 & 1.06 & 1 & 1.08 & 1 & 1.14 & 1 & 1 & 1 & 1 \\
 & 200 & 1.12 & 1 & 1 & 1.07 & 1 & 1.16 & 1 & 1.30 & 1 & 1 & 1 & 1 \\
 & 500 & 1.11 & 1 & 1 & 1.12 & 1 & 1.48 & 1 & 2.02 & 1 & 1 & 1 & 1 \\
 & 1000 & 1.07 & 1 & 1 & 1.12 & 1 & 2.53 & 1 & 3.82 & 1 & 1 & 1 & 1 \\
 & 2000 & 1.07 & 1.05 & 1 & 1.16 & 1.05 & 2.17 & 1 & 5.40 & 1 & 1 & 1 & 1 \\
 & 5000 & 1.07 & 1 & 1 & 1.17 & 1.05 & 4.13 & 1 & 12.85 & 1 & 1 & 1 & 1 \\
 & 10000 & 1.06 & 1 & 1 & 1.19 & 1.06 & 5.92 & 1 & 18.59 & 1 & 1 & 1 & 1 \\
\hline
\noalign{\vskip 1.5pt}
\multirow{8}{*}{\makecell{Multiple\\strongly\\correlated}} & 50 & 1 & 1 & 1 & 1 & 1 & 1 & 1 & 1.05 & 1 & 1 & 1.07 & 1 \\
 & 100 & 1 & 1 & 1 & 1 & 1 & 1 & 1 & 1.08 & 1 & 1 & 1 & 1 \\
 & 200 & 1 & 1.06 & 1 & 1 & 1 & 1 & 1 & 1.05 & 1 & 1 & 1 & 1 \\
 & 500 & 1 & 1.11 & 1 & 1 & 1 & 1 & 1 & 1.08 & 1 & 1 & 1 & 1 \\
 & 1000 & 1 & 1.12 & 1 & 1.05 & 1 & 1.11 & 1 & 1.19 & 1 & 1 & 1 & 1 \\
 & 2000 & 1 & 1.16 & 1 & 1.05 & 1 & 1.21 & 1 & 1.28 & 1 & 1 & 1 & 1 \\
 & 5000 & 1 & 1.19 & 1 & 1.10 & 1 & 1.62 & 1 & 1.80 & 1 & 1 & 1 & 1 \\
 & 10000 & 1 & 1.18 & 1 & 1.13 & 1 & 2.24 & 1 & 2.66 & 1 & 1 & 1 & 1 \\
\hline
\noalign{\vskip 1.5pt}
\multirow{9}{*}{\makecell{Profit\\ceiling}} & 20 & 1.18 & 1 & 1.20 & 1 & 1 & 1 & 1 & 1 & 0.91 & 1 & 1 & 1 \\
 & 50 & 1.15 & 1 & 1.12 & 1 & 1 & 1 & 1 & 1 & 1 & 1 & 1 & 1 \\
 & 100 & 1 & 1 & 1.51 & 1 & 1 & 1 & 1 & 1 & 1 & 1 & 1 & 1.08 \\
 & 200 & 1 & 1.07 & 1.52 & 1 & 1 & 1 & 1 & 1 & 1 & 1 & 1.18 & 1 \\
 & 500 & 1 & 1.12 & 1.66 & 1 & 1 & 1.07 & 1 & 1.09 & 1 & 1 & 1 & 1 \\
 & 1000 & 1 & 1.09 & 1.87 & 1 & 1 & 1.08 & 1 & 1.12 & 1 & 1 & 1 & 1 \\
 & 2000 & 1 & 1.13 & 1.49 & 1 & 1 & 1.18 & 1 & 1.20 & 1 & 1 & 1 & 1 \\
 & 5000 & 1 & 1.11 & 1.85 & 1 & 1 & 1.28 & 1 & 1.31 & 1 & 1 & 1 & 1 \\
 & 10000 & 1 & 1.11 & 1.86 & 1 & 1 & 1.33 & 1 & 1.37 & 1 & 1 & 1 & 1 \\
\hline
\noalign{\vskip 1.5pt}
\multirow{5}{*}{\makecell{Circle}} & 500 & 1 & 1.08 & 1 & 1 & 1 & 1 & 1 & 1 & 1 & 1 & 1 & 1 \\
 & 1000 & 1 & 1.09 & 1 & 1 & 1 & 1.11 & 1 & 1.13 & 1 & 1 & 1 & 1 \\
 & 2000 & 1 & 1.07 & 1 & 1.08 & 1 & 1.18 & 1 & 1.21 & 1 & 1 & 1 & 1 \\
 & 5000 & 1 & 1.13 & 1 & 1.18 & 1 & 1.66 & 1 & 1.76 & 1 & 1 & 1 & 1 \\
 & 10000 & 1 & 1.19 & 1.11 & 1.16 & 1 & 1.99 & 1 & 2.15 & 1 & 1 & 1 & 1 \\
\hline
\noalign{\vskip 1.5pt}
\multirow{5}{*}{\makecell{Jooken}} & 400 & 1.43 & 1 & 1 & 1 & 1 & 1 & 2.90 & 2.92 & 1 & 1 & 1 & 1 \\
 & 600 & 1.52 & 1 & 1 & 1 & 1 & 1 & 3.85 & 3.79 & 1 & 1 & 1.08 & 1 \\
 & 800 & 1.49 & 1 & 1 & 1 & 1 & 1 & 4.02 & 3.99 & 1 & 1 & 1.05 & 1 \\
 & 1000 & 1.51 & 1 & 1 & 1 & 1 & 1 & 3.17 & 3.15 & 1 & 1 & 1 & 1 \\
 & 1200 & 1.60 & 1 & 1 & 1 & 1 & 1 & 4.36 & 4.36 & 1 & 1 & 1.05 & 1.09 \\
\hline
\noalign{\vskip 1.5pt}
\end{tabular}
}
\end{table}

\end{document}